\documentclass[11pt]{article}

\usepackage{graphicx}
\usepackage[usenames,dvipsnames,svgnames,table]{xcolor}
\usepackage[english]{babel}
\usepackage[pagebackref]{hyperref}
\hypersetup{colorlinks=true,urlcolor=blue,linkcolor=blue,citecolor=[rgb]{.0,.50,.32},}
\usepackage{amsmath,amssymb,amsthm,mathrsfs}
\usepackage{mathtools}
\usepackage{enumitem}

\usepackage{fullpage}
\hypersetup{pdfpagemode=UseNone}
\usepackage[margin=1cm]{caption}
\usepackage{thmtools}
\usepackage{thm-restate}
\usepackage[longend,ruled,linesnumbered,nokwfunc]{algorithm2e}
\SetKwInput{Input}{Input}
\SetKwInput{Output}{Output}

\usepackage[capitalize]{cleveref}

\newcommand{\ignore}[1]{}

\renewcommand{\S}{\mathcal{S}}

\newcommand{\R}{\mathbb{R}}

\newcommand{\hw}{\widehat{w}}
\newcommand{\mA}{A}

\newcommand{\mP}{P }
\newcommand{\mSigma}{\mathbf{\Sigma}}
\newcommand{\mDiag}{\mathrm{Diag}}
\newcommand{\mZero}{\mathbf{0}}

\newcommand{\totaliters}{T}
\newcommand{\fobj}{\mathcal{F}}
\newcommand{\primal}{\mathrm{vec}}
\newcommand{\dual}{\mathrm{mat}}
\newcommand{\ones}{\mathbf 1}

\DeclareMathOperator{\Tr}{tr}

\renewcommand{\min}{\mathrm{min}}
\newcommand{\argmin}{\mathrm{argmin}}
\renewcommand{\max}{\mathrm{max}}
\renewcommand{\epsilon}{\varepsilon}
\newcommand{\eps}{\epsilon}

\def\01{\{0,1\}}
\DeclareMathOperator{\polylog}{polylog}

\newtheorem{defin}{Definition}
\newtheorem{definition}[defin]{Definition}
\newtheorem{proposition}[defin]{Proposition}
\newtheorem{theorem}{Theorem}
\newtheorem*{theorem*}{Theorem}
\newtheorem{remark}[defin]{Remark}

\newtheorem{lemma}[defin]{Lemma}

\newtheorem*{claim*}{Claim}

\newtheorem*{conjecture*}{Conjecture}
\theoremstyle{definition}

\DeclareMathOperator{\diag}{diag}
\renewcommand{\vec}{\mathrm{vec}}

\DeclareMathOperator{\logdet}{logdet}

\newcommand{\ma}{A}
\newcommand{\mh}{H}
\newcommand{\mv}{V}
\newcommand{\mw}{W}

\newcommand{\defeq}{:=}
\newcommand{\norm}[1]{\|#1\|}
\newcommand{\normFull}[1]{\left\|#1\right\|}
\newcommand{\tr}{\mathrm{tr}}
\newcommand{\vzero}{\mathbf{0}}

\newcommand{\poly}{\mathrm{poly}}

\newcommand{\Sn}{\S^{n}}

\newcommand{\sigmaipw}{\sigma_{i}(\mw^{\frac{1}{2}-\frac{1}{p}}\ma)}
\newcommand{\sigmapw}{\sigma(\mw^{\frac12-\frac1p} \ma)}

\let\oldnl\nl
\renewcommand{\d}{\mathrm{d}}
\newcommand{\nonl}{\renewcommand{\nl}{\let\nl\oldnl}}
\newcommand{\intC}{\mathrm{int}\,\mathcal{C}}
\newcommand{\alphap}{\alpha_{p}}
\newcommand{\betap}{\beta_{p}}
\newcommand{\baralphap}{{\bar{\alpha}_{p}}}
\newcommand{\barbetap}{{\bar{\beta}_{p}}}

\newcommand\blfootnote[1]{%
  \begingroup
  \renewcommand\thefootnote{}\footnote{#1}%
  \addtocounter{footnote}{-1}%
  \endgroup
}

\title{Computing Lewis weights to high precision\\ using local relative smoothness}

\author{
Sander Gribling\thanks{Tilburg University, \texttt{s.j.gribling@tilburguniversity.edu}}
\and
Aaron Sidford\thanks{Stanford University, \texttt{\string{sidford,chenyiz\string}@stanford.edu}}
\and
Chenyi Zhang$^\dagger$
}

\date{}

\begin{document}

\maketitle

\begin{abstract}
We provide algorithms that compute $\epsilon$-estimates of the $\ell_p$-Lewis weights of a matrix $A \in \R^{m \times n}$ for $p \geq 4$ using $O(p^2 \log(m/\epsilon))$ rounds of leverage score computation, where $\ell_p$-Lewis weights and leverage scores are both standard measures of row importance. This improves upon the state-of-the-art round complexity of $O(p^3 \log(m/\epsilon))$ due to Fazel, Lee, Padmanabha, and Sidford (2022). We obtain our results by carefully applying a local variant of relatively smooth gradient descent to primal and dual forms of the $\ell_p$-Lewis weight optimization problem and providing tools to convert between different notions of approximate $\ell_p$-Lewis weights.\blfootnote{This work subsumes the note ``On computing approximate Lewis weights'' by Apers, Gribling, Sidford~\cite{apers2024lewis}.} 
\end{abstract}

\section{Introduction} 

The \emph{$\ell_p$-Lewis weights of a matrix $A \in \R^{m \times n}$
}, denoted $\sigma_p(A) \in \R^{m}_{\geq 0}$, are a fundamental measure of the importance of the rows of $A$~\cite{Lewis78,BLM89,cohen2015lp}. They arise in sampling schemes for sparsifying a matrix with respect to the $\ell_p$-norm~\cite{cohen2015lp,lee16,JLS22}, self-concordant barriers for linear programming~\cite{lee2019solving,BLLSSWW21}, optimal design problems in statistics, and geometric problems \cite{todd2016minimum,cohen2015lp,fazel22}. 

The $\ell_p$-Lewis weights can be viewed as an $\ell_p$-generalization of the \emph{leverage scores of $A$}, denoted $\sigma(A)$, a natural measure of the importance of a row in $\ell_2$. Leverage scores have many applications in statistics \cite{rudelson2007sampling}, randomized linear algebra \cite{drineas2012fast}, and graph algorithms \cite{spielman2008graph}.
In the case that $A$ has full (row-)rank (which we assume for simplicity), $\sigma(A)$ is defined as
\begin{equation*}
\sigma(A)_i \defeq a_i^\top\bigl(A^\top A\bigr)^{-1}a_i
\text{ and }
a_i^\top
\text{ is the $i$'th row of $A$ for all $i \in [m]$\,.
}
\end{equation*}
The $\ell_p$-Lewis weights of $A$ are the leverage scores of $A$ when $p = 2$. Otherwise, they are defined implicitly as the leverage scores after appropriate re-weighting the rows by these same leverage scores. Below we define them formally for what we call \emph{non-degenerate} matrices.\footnote{This is a mild restriction as zero rows have both  a leverage score and a Lewis weight of zero and rank-deficiency can be removed by carefully restricting to a smaller subspace or analyzed by replacing matrix inverses with pseudoinverses.} 

\begin{definition}[$\ell_p$-Lewis weights] 
For $p \in (0, \infty)$ the $\ell_p$-Lewis weights of non-degenerate, i.e., full-rank with no zero rows, $A\in\mathbb{R}^{m\times n}$, denoted $\sigma_p(A) \in \R^{m}_{> 0}$,
is the unique \cite{Lewis78,wojtaszczyk1991,cohen2015lp} 
positive $w \in \R^{m}_{> 0}$ where, for $W = \mDiag(w)$, $w = \sigma(W^{\frac{1}{2} - \frac{1}{p}} A)$.
\end{definition}

To further motivate $\ell_p$-Lewis weights, we consider two fundamental problems from statistics and geometry associated with a set of vectors $\{a_i\}_{i \in [m]} \subseteq \R^n$, see e.g., \cite{khachiyan1996rounding,todd2016minimum}.\footnote{We assume for simplicity that $\{a_i\}_{i \in [m]}$ is centrally symmetric, that is,  $\{a_i\}_{i \in [m]} = - \{a_i\}_{i \in [m]}$.} First, the $D$-optimal design problem in statistics associates the $\{a_i\}_{i \in [m]}$ to experiments and asks to assign probabilities $\lambda \in \R^m_{\geq 0}$ to them in order to minimize the determinant of the error covariance matrix $(\sum_{i \in [m]} \lambda_i a_i a_i^\top )^{-1}$, i.e., 
\begin{equation} \label{eq:Doptimal}
\min_{\lambda \in \R^m_{\geq 0}} -\logdet\Big(\sum_{i\in[m]} \lambda_i a_i a_i^\top\Big) \enspace \text{ s.t. } \enspace \ones^\top \lambda =1.
\end{equation}
This corresponds to designing the experiment to minimize the volume of the resulting confidence ellipsoid (for any fixed confidence level). The (Lagrange) dual problem,
\begin{equation} \label{eq:John}
\min_{M \in \Sn_{\succ 0}} -\logdet(M) \enspace \text{ s.t. } \enspace a_i^\top M a_i \leq 1 \enspace \forall i \in [m].
\end{equation}
is a fundamental geometric problem that computes the minimum volume ellipsoid  $\mathcal E \defeq \{x \in \R^n : x^\top M x \leq 1\}$ that contains the vectors $\{a_i\}_{i \in [m]}$, i.e., the John ellipsoid of the set $\{a_i\}_{i \in[m]}$~\cite{john1948}. Cohen and Peng~\cite{cohen2015lp} showed that an $\ell_p$-variant of \eqref{eq:John} is connected to Lewis weights: based on \cite{wojtaszczyk1991}, they introduced the convex program 
\begin{align} \label{eq:convexprog} 
\min_{M \in \Sn_{\succ 0}} -\logdet(M) \enspace \text{ s.t. } \enspace \sum_{i\in [m]} (a_i^\top M a_i)^{p/2} \leq n\,.
\end{align}
The optimal solution $M_{*}$ of \eqref{eq:convexprog} satisfies $M_{*}^{-1} = \sum_{i \in [m]} (a_i^\top M_{*} a_i)^{p/2-1} a_i a_i^\top$ and therefore encodes the Lewis weights of $A$ as $\sigma_p(A)_i = (a_i^\top M_{*} a_i)^{p/2}$.

Given their applications and connections, computing $\ell_p$-Lewis weights is a prominent structured optimization problem. Additionally, given the well-studied nature of leverage scores and their simple expression linear-algebraically, previous work studied how much more algorithmically challenging it is to compute Lewis weights \cite{cohen2015lp,lee16,lee2019solving,fazel22}. Similarly, in this paper we study the following central question:
\begin{center} \emph{How many leverage score computations, i.e., computing $\sigma(D A)$\\ for diagonal $D$,
suffice to estimate $\sigma_p(A)$?}
\end{center} 
We seek new algorithmic and analytic tools for answering this question.

\subsection{Prior work}
For $p \in (0, 4)$, it is straightforward to show the map $m(w) \defeq \sigma(\mDiag(w)^{\frac{1}{2} - \frac{1}{p}} A)$ is multiplicatively contractive for $w \in \R^{m}_{> 0}$~\cite{cohen2015lp}. Iteratively applying $m$ gives an algorithm that computes an $\epsilon$-estimate of $\sigma_p$, i.e., $\hat{w} \in \R^{m}_{> 0}$ with $(1 - \epsilon) \sigma_p(A) \leq \hat{w} \leq (1 + \epsilon) \sigma_p(A)$ using $O(\frac{1}{1-|1-p/2|} \cdot \log(\log(m/\eps)))$ leverage score computations \cite{cohen2015lp}. 

However, efficiently obtaining $\epsilon$-estimates of $\sigma_p(A)$ for $p \geq 4$ has been more challenging. Cohen and Peng~\cite{cohen2015lp} showed that, by applying the ellipsoid method to \eqref{eq:convexprog}, estimates can be computed in $O(m\cdot \poly(n)\log(1/\eps))$ \textit{time}. They also provided a recursive algorithm to compute high-accuracy estimates using $\Omega(n)$ leverage score computations. Additionally, Lee and Sidford~\cite{lee2019solving} show how to compute $\eps$-estimates using $O(\sqrt{n} \cdot p^2 \polylog(mn/\eps))$ leverage score computations. Their approach applies a descent method to a volumetric potential (equivalent to $\fobj_\primal$ defined later up to a change of coordinates)  that captures $\ell_p$-Lewis weights. They show that the Hessian is stable around the minimizer which makes the convex objective function \textit{locally} well conditioned. This ensures a $\log(1/\eps)$-dependence of the descent method once weights are found that are close enough to the minimizer. To find such initial weights they used a homotopy method that slowly increases~$p$. 

Only recently, Fazel, Lee, Padmanabhan, and Sidford~\cite{fazel22} provided the only known algorithms which compute $\epsilon$-estimates of Lewis weights for $p > 2$ using a \emph{nearly dimension-free} number of leverage score computations. Their method used $O(p^3 \log(mp/\eps))$ leverage score computations. The derivation and analysis of their algorithms leveraged the following convex optimization problem, where we let $V \defeq \diag(v)$,
\begin{align}
\label{eq:vec-opt-problem}
    \min_{v \in \R^{m}_{> 0}} \fobj_{\primal}(v)
    \text{ where }
    \fobj_{\primal}(v)\defeq -\log \det ( A ^\top  V A ) + \frac{1}{1+\alphap}\ones^\top v^{1+\alphap}\text{ for }\alphap \defeq \frac2{p-2}.
\end{align}
Optimality conditions imply that its minimizer $v_{*}$ satisfies $[v_{*}]_i^{\alphap} = a_i^\top (A^\top V A)^{-1} a_i$, and therefore $v_*^{1+\alphap} = \sigma_p(A)$.\footnote{This rescaling of the coordinates of the Lewis weights to the $\frac{1}{1+\alphap} = 1-\frac2p$ power is often convenient to work with and we use $v$ rather than $w$ to indicate vectors in this rescaled space.} Fazel et al.~\cite{fazel22} departed from contractivity analysis and instead performed an innovative, seemingly bespoke, analysis of \eqref{eq:vec-opt-problem}. 

The key insight of~\cite{fazel22} is that a type of quasi-Newton step significantly decreases $\fobj_{\primal}(v)$ when a geometrically motivated invariant holds. The invariant is $\rho_\max(v) \leq 1+\alphap$, where 
\begin{equation*}
    \rho_\max(v) \defeq \max_{i \in [m]} \, \rho_i(v)
    \enspace\text{where}\enspace
    \rho_i(v) \defeq \frac{a_i^\top (A^\top V A)^{-1} a_i}{v_i^{\alphap}}
    =
    \frac{\sigma_i(V^{\frac12}A)}{v_i^{1 + \alphap}}
    \enspace\text{for all}\enspace i \in [m].
\end{equation*}
Note that $\rho_i(v_{*})=1$ for all $i \in [m]$, which means that the distance from $\rho(v)$ to the all-ones vector is a proxy for closeness to (rescaled) Lewis weights. The quantity $\rho_\max(v)$ has a geometric interpretation:  
$\{x \in \R^n : x^\top A^\top V A x\leq 1 \} \subseteq \{x \in \R^n: \|V^{-\alphap/2} A x\|_\infty \leq \sqrt{\rho_\max(v)}\}$, which can be viewed as a notion of rounding~\cite{fazel22}.\footnote{Indeed, $a_i^\top (A^\top V A)^{-1} a_i \leq \rho_\max(v) v_i^{-\alphap}$. Hence, the rescaled vectors $\{ a_i/(v_i^{\alphap/2}\sqrt{\rho_\max(v)})\}_{i \in [m]}$ belong to the ellipsoid $\{x \in \R^n: x^\top  (A^\top VA)^{-1} x \leq 1\}$. The statement follows by considering the polar of each set.}

To ensure that the geometric invariant is maintained, Fazel et al.~\cite{fazel22} introduced a rounding procedure and provided an algorithm that uses $O(p^3 \log(mp/\eps))$ leverage score computations and alternates between applying the rounding procedure and applying the quasi-Newton step. The quasi-Newton step can be written as updating $v$ to $v^+$ where $v_i^+ = \left(1+\eta \frac{\rho_i(v)-1}{\rho_i(v)+1}\right) v_i$, and the step-size $\eta$ is $1/3$ for $p \geq 4$. Additionally, they provided another algorithm which avoids the rounding procedure by varying the step-size $\eta$ per coordinate $i$ depending on whether $\rho_i(v) \geq 1$ or $\rho_i(v)<1$; it also uses $O(p^3 \log(mp/\eps))$ leverage score computations. 

There are additional algorithms that compute weaker approximations than $\epsilon$-estimates of $\ell_p$-Lewis weights \cite{cohen2015lp,lee16}. To motivate these notions, recall that if $w$ is an $\ell_p$-Lewis weight vector, then it satisfies the fixed-point equation $w = \sigma(W^{\frac12-\frac1p}A)$, and therefore $\|w\|_1 = \|\sigma(W^{\frac12-\frac1p} A)\|_1 = n$. By relaxing the fixed-point equation to a one- or two-sided inequality, we arrive at the following (increasingly strong up to constants depending on $p$) notions of approximate $\ell_p$-Lewis weights.

\begin{definition}[Lewis weight approximations]
Let $\ma \in \R^{m \times n}$ be a non-degenerate, $w \in \R^m_{> 0}$, $0< \eps <1$, and $p >0$. Then we say
\begin{itemize}
\item
$w$ is 
a \emph{one-sided $\epsilon$-approximation} of $\sigma_p(A)$ if $\sigma(W^{\frac12-\frac1p}A) \leq (1+\eps) w$ and $\|w\|_1 \leq (1+\eps)n$. 
\item 
$w$ is a \emph{two-sided $\epsilon$-approximation} of $\sigma_p(A)$ if $(1-\eps)  \sigma(W^{\frac12-\frac1p}A) \leq w \leq (1+\eps) \sigma(W^{\frac12-\frac1p}A)$.
\item
$w$ is an \emph{$\eps$-estimate} of $\sigma_p(A)$  if $(1-\eps)\sigma_p(A) \leq w \leq (1+\eps) \sigma_p(A)$.
\end{itemize}
\end{definition}
For many $\ell_p$-embedding and -regression problems, the weakest one-sided approximation suffices, even when $\|w\|_1 = O(d)$, see~\cite{talagrand1990embedding,cohen2015lp,Woodruff2023online}. Lee showed that iteratively applying the map $m(w)$ for $T = O(\log(m/n)/\eps)$ iterations and outputting the average of the iterates results in a one-sided $\eps$-approximation~\cite[Theorem~5.3.4]{lee16}. For some applications in optimization, however, a stronger notion of estimates are used~\cite{lee2019solving,apers2023quantum-journal}. A natural question is how the various notions are related to each other. The only previously known conversion is that a two-sided $\eps$-approximation is also an $O(\eps p^2 \sqrt{n})$-estimate~\cite[Lemma 14]{fazel22}.

\subsection{Our results}

In this paper we develop two new algorithms for computing $\epsilon$-estimates of $\ell_p$-Lewis weights for $p > 2$. Our algorithms use only $O(p^2\log(m/\epsilon))$ leverage score computations, improving upon the prior nearly-dimension free results by a factor of $p$.  (See Table~\ref{tab:summary}.)

\begin{table}[ht]
\small
\centering
\begin{tabular}{ccc}
\hline
\# Computes  & Optimality & Reference \\ \hline \hline
$O(p^3\log(mp/\epsilon))$  & $\epsilon$-estimate & \cite{fazel22} \\ \hline
 $O(\log(m/n)/\eps)$ & $\epsilon$-one-sided & \cite{lee16} \\ \hline
$O(p^2\log(m/\epsilon))$  & $\epsilon$-estimate & Algorithm~\ref{alg:relative-smoothness-dual} \\ \hline
$O(p^2\log(mp/\epsilon))$ & $\epsilon$-estimate & Algorithm~\ref{alg:relative-smoothness} \\ \hline
\end{tabular}
\caption{Comparison between the prior state of the art and our work, for the regime $p \geq 4$. The number of computes measures the number of leverage score computations. } 
\label{tab:summary}
\end{table}

Moreover, we show how to obtain these results by a fairly straightforward algorithm  (the complete pseudocode is given later in Algorithm~\ref{alg:relative-smoothness-dual}): starting from the all-ones vector $v^{(0)} = \ones$, it performs the following iteration $T=O(p^2\log(m/\epsilon))$ many times 
\begin{align} 
\label{eq:algintro}
\framebox{\quad \rule[-1.2em]{0pt}{3em} $v_i^{(t+1)} = \Big(1+ \frac{\rho_i(v^{(t)})^{1/\alphap}-1}{L} \Big)v_i^{(t)},  \qquad \forall i \in [m]$ \quad } 
\end{align}
where $L$ is a suitably chosen step-size, and outputs $\widehat w = \hat v^{1+\alphap}$ for $\hat v = (a_i^\top (A^\top V^{(T)} A)^{-1} a_i)^{1/\alphap}$. 

\begin{restatable}{theorem}{theomat} \label{thm:relative-smoothness-dual-main}
For $p>2$,  \Cref{alg:relative-smoothness-dual} outputs an $\epsilon$-estimate of $\sigma_p(A)$ in $O(p^2\log(mp\alphap/\epsilon))$ iterations. Each iteration computes the leverage scores of $DA$ of some diagonal matrix $D$.
\end{restatable}

Together with \cite{cohen2015lp}, Theorem~\ref{thm:relative-smoothness-dual-main} gives the state-of-the-art rates for computing the $\ell_p$-Lewis weights for all regimes of $p$.

Excitingly, rather than a particularly tailored analysis of a potential function, we analyze this algorithm using \emph{relative smoothness and relative strong-convexity}~\cite{lu2018relatively}, which are general regularity assumptions used in analyzing gradient-based methods for convex optimization, see also \cite{bauschke2017,tseng2008accelerated}. Via a simple extension, we prove that relatively smooth gradient descent converges at rates similar to those established in \cite{lu2018relatively} even when only a local variant of relative smoothness holds. We show how \eqref{eq:algintro} is essentially equivalent to applying this method to a suitable objective. Additionally, we show that the convergence guarantees of this method directly correspond to computing $\epsilon$-estimates of Lewis weights.

Complementing this result, we show that relative smoothness can also be applied directly to~\eqref{eq:vec-opt-problem}, the optimization problem considered in \cite{fazel22}. We show that replacing \eqref{eq:algintro} with 
\begin{align} 
\framebox{\quad \rule[-1.2em]{0pt}{3em}  $v_i^{(t+1)}\leftarrow\Big(1+\frac{\rho_i(v^{(t)})-1}{L}\Big)^{1/\alphap}v_i^{(t)},\quad\forall i\in[m]$ \label{lin:fobj-rs-update-intro} \quad }
\end{align}
for suitably chosen $L$ optimizes $\fobj_{\primal}(v)$ to accuracy $\epsilon > 0$ in $O(p^2\log(mp^2\alphap/\epsilon))$ iterations.  
However, as in \cite{fazel22}, significant work is needed to convert the convergence in function value to a guarantee on closeness to Lewis weights. We later provide \Cref{alg:relative-smoothness} which does this and analyze it in several steps. First, we show that the iterates from~\eqref{lin:fobj-rs-update-intro} in fact converge to one-sided approximations. 

\begin{restatable}{theorem}{theoapproxone}
\label{thm:primal_algorithm_one_sided_guarantee}
For $p>2$, \Cref{alg:relative-smoothness} with parameter $\hat\epsilon$ produces, after $T=O(p^2\log(mp^2\alphap/\hat\epsilon))$ iterations, a vector $w\coloneqq [v^{(T)}]^{1+\alphap}$ that is a one-sided $\hat\epsilon$-approximation of $\sigma_p(A)$.
\end{restatable}

Then, we establish two new results that show how to convert a one-sided approximation to either a two-sided approximation or a multiplicative estimate, where $\barbetap\coloneqq \max\{1,1/\alphap\}=\max\big\{1,\frac{p-2}{2}\big\}$.
\begin{restatable}{theorem}{theoonetotwo}\label{thm:one-sided-to-two-sided} 
For $p\geq 2$, if $w$ is a one-sided $\epsilon_{\mathrm{one}}$-approximation of $\sigma_p(A)$ and $\widehat w\coloneqq \sigma(w^{\frac12-\frac1p})^{\frac p2}/w^{\betap}$, then $\widehat w$ is a two-sided $\epsilon_{\mathrm{two}}$-approximation of $\sigma_p(A)$ for $\epsilon_{\mathrm{two}}=3\barbetap n\epsilon_{\mathrm{one}}(1+\epsilon_{\mathrm{one}})^{\barbetap}$.
\end{restatable}

\begin{restatable}{theorem}{theoonetomulti}\label{thm:one-sided-to-multiplicative}
For $p>2$, suppose $w$ is a one-sided $\epsilon_{\mathrm{one}}$-approximation of $\sigma_p(A)$ satisfying
\[
\epsilon_{\mathrm{one}}\leq \frac{1}{\barbetap n}\min\left\{\frac{1}{96(p-2)^2(4p-7)^2},\frac1{50}\right\}.
\]
Define $\widehat w \in \mathbb R^m_{>0}$ by $\widehat w_i \coloneqq \sigma_i(w^{\frac12-\frac1p})^{\frac p2}/w_i^{\betap}$ for each $i \in [m]$. Then $\widehat w$ is an $\epsilon_{\mathrm{est}}$-estimate of $\sigma_p(A)$, where $\epsilon_{\mathrm{est}}=2(p-2)(4p-7)\sqrt{6\barbetap n\epsilon_{\mathrm{one}}}$.
\end{restatable}

Applying the postprocessing from Theorem~\ref{thm:one-sided-to-multiplicative} to the final iterate of \Cref{alg:relative-smoothness} yields an $\epsilon$-estimate as reflected in the following \Cref{thm:primal-algorithm-guarantee}.

\begin{restatable}{theorem}{theovec}\label{thm:primal-algorithm-guarantee}
For $p>2$, Algorithm~\ref{alg:relative-smoothness} outputs an $\epsilon$-estimate of $\sigma_p(A)$ in $O(p^2\log(mp^2\alphap/\epsilon))$ iterations. Each iteration computes the leverage scores of $DA$ of some diagonal matrix $D$.
\end{restatable}

Additionally, we establish two new results using our conversion tools. First, in \Cref{thm:approx-min-to-two-sided} we give a postprocessing step that transforms any approximate minimizer $v$ of $\fobj_\primal$ satisfying $\rho_\max(v) \leq 1+\epsilon$ into a two-sided approximation. Compared to the postprocessing step in \cite[Lemma~1]{fazel22}, our approach does not incur a dimension-dependent polynomial factor loss in accuracy. Second, in \Cref{sec:improved-analysis-Lee's-alg} we provide an improved analysis of a variant of \cite[Algorithm~6]{lee16}, obtaining two-sided $\epsilon$-approximations from $O(pn\log m/\epsilon)$ approximate leverage-score computations to accuracy $O(\epsilon/(pn))$.

\begin{restatable}{theorem}{theoapproxtotwo} \label{thm:approx-min-to-two-sided} 
For $p>2$ and $\epsilon\leq\min\big\{\frac1{1000},\frac1{50\baralphap}\big\}$, suppose $v\in\mathbb R^m_{>0}$ satisfies $\fobj_\primal(v)-\fobj_\primal(v_*)\leq \epsilon^3$ and $\rho_{\max}(v)\leq 1+\epsilon$. Define $\widetilde w\in\mathbb R^m_{>0}$ coordinatewise by setting $\widetilde w_i=(\sigma_i(v)/v_i)^{1+1/\alphap}$ if $\rho_i(v)\leq 1-\epsilon$, and $\widetilde w_i=v_i^{1+\alphap}$ otherwise. Then $\widetilde w$ is a two-sided $50\max\{\alphap,1\}\epsilon$-approximation of $\sigma_p(A)$. 
\end{restatable}

Though not the main focus of our work, we briefly discuss the runtime of our algorithms due to leverage score computations. Exact leverage scores of $DA$ can be computed by first computing $G = A^\top D^2 A$ in time $O(m n^{\omega-1})$, then computing $H = G^{-1} A^\top D$ in time $O(m n^{\omega-1})$, and then computing the inner product of column $i$ of $H$ with row $i$ of $DA$ in time $O(mn)$ for all $i$. To the best of our knowledge, there is no better runtime to compute the leverage scores to high precision, though faster randomized algorithms for approximately computing leverage scores are known~\cite{spielman2008graph,clarkson2017low}. The conditioning of $D$ affects this procedure through the required bit precision. Hence, we view controlling the range of $D$ as an interesting open problem. We note that in both our algorithms, the diagonal scaling $D^{(t)}=(V^{(t)})^{1/2}$ changes by only a constant multiplicative factor in each coordinate between consecutive iterations, see Remarks~\ref{rem:dual-coordinatewise-scaling-change}, \ref{rem:primal-coordinatewise-scaling-change}.

\subsection{Approach}

Here we provide a brief overview of our approach. First, we briefly sketch the relative smoothness and convexity framework. (Formal definitions are deferred to \cref{sec:relative_smoothness_framework}.) For differentiable functions $f$ and $h$, we say that $f$ is $\mu$-strongly convex and $L$-smooth relative to $h$ when 
\begin{equation}
\label{eq:intro-rel}
\mu D_h(x,y)  \leq f(x)-f(y)-\langle \nabla f(y),x-y\rangle \leq L D_h(x,y) \qquad \forall x,y, 
\end{equation}
where $D_h(x,y) \defeq h(x)-h(y) - \langle \nabla h(y),x-y\rangle$ is the Bregman divergence associated with $h$. For twice-differentiable $f$ and $h$, Lu, Freund, and Nesterov~\cite{lu2018relatively} show that \eqref{eq:intro-rel} is equivalent to $\mu \nabla^2 h(x) \preceq \nabla^2 f(x) \preceq L \nabla^2 h(x)$ for all $x$. They also show, roughly, that when $0<\mu<L$, the gradient descent scheme 
\begin{align*}\label{eqn:relative-smoothness-update}
x^{(t+1)}\leftarrow\argmin_{x\in \mathcal{C}}\left\{f(x^{(t)})+\langle\nabla f(x^{(t)}),x-x^{(t)}\rangle+LD_h(x,x^{(t)})\right\}
\end{align*}
converges linearly at a rate $1-\frac{\mu}{L}$ in function value, and in Bregman distance to the minimizer, e.g., $D_h(x^*,x^{(t)})\leq \frac L\mu\left(1-\frac\mu L\right)^{t}D_h(x^*,x^{(0)})$ where $x^*$ is the minimizer of $f$.

In a nutshell, the update step in \eqref{eq:algintro} arises from a careful extension of this framework applied to the convex optimization problem over positive definite $n$-by-$n$ matrices  
\begin{equation} 
\begin{aligned}
\label{eq:mat-opt-problem}
     &\min_{M \in \Sn_{\succ 0}} \, 
    \fobj_{\dual}(M), \text{ where}\\
    &
    \fobj_{\dual}(M) \defeq -\logdet(M) + \frac{1}{1 + \betap} \sum_{i \in[m]} (a_i^\top M  a_i)^{1+\betap}
    \text{ for }
    \betap \defeq \frac{p - 2}{2}\,. 
    \end{aligned}
\end{equation}
This problem is essentially dual to \eqref{eq:vec-opt-problem}; see Lemma~\ref{lem:duality} and note that $\betap = \alphap^{-1}$. The minimizer $M_{*}$ of \eqref{eq:mat-opt-problem} satisfies $M_{*}^{-1} = A^\top \overline V A$, where $\overline V\defeq \sigma_p(A)^{\frac1{1+\alphap}}$, and thus encodes the Lewis weights. The update from \eqref{eq:algintro} corresponds to the gradient descent method applied the above objective where $M^{-1} = A^\top V A$. It is easy to see that $\fobj_\dual$ is $1$-strongly convex relative to $h_\dual(M) = -\logdet(M)$. Moreover, the Bregman divergence associated to $h_\dual$ roughly measures spectral closeness between matrices: $D_{h_\dual}(M_{*},M) \leq \eps^2/4$ implies $(1-\eps) M_{*} \preceq M \preceq (1+\eps) M_{*}$ (Lemma~\ref{lem:bregman_divergence_to_spectral_difference}), which shows that near-optimal points of $\fobj_\dual$ provide $\eps$-estimates of the Lewis weights.

The only remaining challenge is to establish the relative smoothness of $\fobj_\dual$ with respect to $h_\dual$. Unfortunately, a sufficient global bound is unknown, even for sub-level sets. Instead, we introduce \textit{local relative smoothness between iterates}, a straightforward extension of relative smoothness that just holds between the iterates. The idea of using different local (or adaptive) notions of smoothness has been used before, e.g., in \cite{sidford2018coordinate,malitsky2020adaptive,latafat2025}. In particular, Li et al.~\cite{li2018general} developed a ball-local version of relative smoothness, requiring the relative-smoothness inequality to hold uniformly within a neighborhood of each point, while Godeme et al.~\cite{godeme2023provable} used local relative strong convexity on a prescribed neighborhood, typically around a solution.

We prove in \cref{sec:relative_smoothness_framework} that such local relative smoothness suffices for linear convergence. In particular, we show that when $M = (A^\top V A)^{-1}$ for some $V = \mDiag(v)$ with $v \in \R^m_{>0}$, then
\[
\nabla^2\fobj_\dual(M ) \preceq \left(1+\betap \Phi_\max(v)\right) \nabla^2 h_\dual(M),
\]
where 
\begin{equation*}
    \Phi_\max(v) \defeq \max_{i \in [m]} \, \Phi_i(v)
    \text{ where }
    \Phi_i(v) \defeq \frac{(a_i^\top (A^\top V A)^{-1} a_i)^{\betap}}{v_i}
    =
    \rho_i(v)^{1/\alphap}
    \text{ for all } i \in [m]\,.
\end{equation*}
Applying this gradient descent scheme to $\fobj_\dual$ and $h_\dual$ results in the iterates \eqref{eq:algintro}, when written in terms of $v$. To establish convergence, we set $L = 32p\,\max\{\betap,1\}$ and show that $\Phi_\max(v)$ is uniformly bounded in the segment between each pair of iterates $v^{(t)}$ and $v^{(t+1)}$. 

As discussed, to further showcase the approach, we then apply the same local relative smoothness framework to the potential $\fobj_\primal$ used in \cite{lee2019solving,fazel22}. In this case, it is easy to see that $\fobj_\primal$ is $1$-strongly convex relative to $h_\primal(v) = \frac{1}{1+\alphap} \ones^\top v^{1+\alphap}$. We show that for any $v \in \R^m_{>0}$, 
\[
\nabla^2 \fobj_\primal(v) \preceq \left(1 + \alphap^{-1} \rho_{\max}(v) \right) \nabla^2 h_\primal(v).
\]
In a similar fashion as for \eqref{eq:algintro}, we establish \textit{local} relative smoothness between iterates when $L=32p/\alphap$, thus establishing convergence in function value. As discussed earlier, with more work we are able to use this to obtain approximations of Lewis weights. One component of this reduction is an efficient conversion of one-sided approximations into two-sided approximations (see \cref{sec:one_sided_to_others}). We do so by defining a transformation $w \mapsto \widehat w$ such that $\|\rho(\widehat w^{\frac12-\frac1p})-1\|_\infty$ is controlled by the spectral approximation quality of $A^\top W^{\frac12-\frac1p}A$ by $A^\top \widehat W^{\frac12-\frac1p} A$, which can in turn be bounded by $\|\widehat w - \sigma(w^{\frac{1}{2}-\frac1p})\|_1$ and is small whenever $w$ is a one-sided approximation.

\subsection{Preliminaries: paper organization and notation}

In the remainder of this paper we extend the relative smoothness framework (\cref{sec:relative_smoothness_framework}), apply it to $\fobj_\dual$ (\cref{sec:matrix_function_definition}) and $\fobj_\primal$ (\cref{sec:applying_rs_to_primal}), and establish conversion results between the notions of approximation (\cref{sec:one_sided_to_others}). We conclude this introduction by providing our notation.

\paragraph{General notation.}
We use lowercase letters for vectors and capital letters for matrices. When the context is clear, a capital letter additionally denotes the diagonal matrix formed from its lowercase counterpart, e.g., $W=\mDiag(w)$. The all-zero and all-one vectors of appropriate dimension are denoted by $\vzero,\ones$, respectively. For $u,v\in\R^n$, we write $u\le v$ to denote the entrywise inequality $u_i\le v_i$ for all $i\in[n]$, and write $u\approx_\epsilon v$ if $(1-\epsilon)u_i\leq v_i\leq (1+\epsilon)u_i$ for all $i\in[n]$. We use $\mDiag(u)$ to denote the diagonal matrix with entries $\mDiag(u)_{ii}=u_i$. For any matrices $A, B$, we write $A \succeq B$, or equivalently, $B \preceq A$ when $A-B$ is positive semidefinite. Moreover, we define $\langle A,B\rangle\coloneqq \Tr(A^\top B)$.
We write $\Sn$ for the space of symmetric $n$-by-$n$ matrices. For a matrix $A \in \Sn$ we use $\norm{A}$ and $\norm{A}_1$ to denote its spectral norm and Schatten $1$-norm respectively. For any convex set $\mathcal{C}$, we use $\intC$ to denote its interior. We use $\otimes$ to denote the Kronecker product.

\paragraph{Lewis weight notation.}
For a matrix $A\in\R^{m\times n}$, we write $v_{*}(A)\coloneqq\sigma_p(A)^{\frac{1}{1+\alphap}}$ when $p$ is clear from context, and write $v_{*}$ when the underlying matrix is clear from the context. Denote $V_*=\diag(v_*)$. For any $p>2$, we denote $\alphap=\frac{2}{p-2}$, $\betap=1/\alphap$, $\baralphap=\max\{1,\alphap\}$, and $\barbetap=\max\{1,\betap\}$.

\section{Locally relatively smooth gradient descent framework}\label{sec:relative_smoothness_framework}

In this section, we present a straightforward local extension of the relative smoothness framework introduced in~\cite{lu2018relatively}, where the relative smoothness condition only holds locally, along linear combinations of selected pairs of points. This differs from the local relative smoothness framework of~\cite{li2018general}, in which relative smoothness holds when restricted to a ball around any given point.

\begin{definition}[Local relative smoothness]  
Let $f,h:\mathcal C\to\R$ be differentiable functions on a convex set $\mathcal C$, and let $x,y\in\intC$. We say that $f$ is $L$-smooth relative to $h$ between $x,y$ if we have
\[
f((1-\lambda)x+\lambda y)\leq f(x)+(1-\lambda)\langle\nabla f(x),y-x\rangle+LD_h((1-\lambda)x+\lambda y,x),\quad\forall \lambda\in[0,1].
\]
If $f$ and $h$ are twice differentiable, this condition is equivalent to 
\[
\nabla^2 f((1-\lambda)x+\lambda y) \preceq L\nabla^2 h((1-\lambda)x+\lambda y),\quad\forall \lambda\in[0,1].
\]
\end{definition}

We show that for any objective function $f$ defined on $\mathcal{C}$ that is $\mu$-strongly convex relative to some known convex function $h$, repeatedly performing the following update
\begin{align}\label{eqn:relative-smoothness-update}
x^{(t+1)}\leftarrow\argmin_{x\in \mathcal{C}}\left\{f(x^{(t)})+\langle\nabla f(x^{(t)}),x-x^{(t)}\rangle+LD_h(x,x^{(t)})\right\}
\end{align}
converges to the minimizer of $f$, given that $f$ is $L$-smooth relative to $h$ between each $x^{(t)}$ and $x^{(t+1)}$.

\begin{proposition}\label{prop:primal-gradient-scheme-convergence-local}
Let $f,h:\mathcal C\to\R$ be differentiable functions on a convex set $\mathcal C$ where $f$ is $\mu$-strongly convex relative to $h$ for some $\mu \geq 0$, and $h$ is convex. If in the updating scheme~(\ref{eqn:relative-smoothness-update})
there exists $L>0$ such that $f$ is $L$-smooth relative to $h$ between $x^{(t)}$ and $x^{(t+1)}$ for every iteration $t$, then
\begin{align*}
D_h(x,x^{(t)})+\frac1L\sum_{k\in[t]}\left(1-\frac\mu L\right)^{t-k}(f(x^{(k)})-f(x))\leq \left(1-\frac\mu L\right)^{t}D_h(x,x^{(0)}),\quad\forall x\in \mathcal C,t\in \mathbb{N}^*.
\end{align*}
Consequently, for $x^*\coloneqq\argmin_{x\in\mathcal C}f(x)$ and when $\mu>0$,
\[
D_h(x^*,x^{(t)})\leq \left(1-\frac\mu L\right)^{t}D_h(x^*,x^{(0)})
\ \text{  and  } \ 
f
(x^{(t)})-f(x^*)\leq \frac{\mu(1-\mu/L)^t}{1-(1-\mu/L)^t}\cdot D_h(x^*,x^{(0)}).
\]

\end{proposition}
The proof of Proposition~\ref{prop:primal-gradient-scheme-convergence-local} is inspired by Theorem 3.1 of \cite{lu2018relatively}. Our contribution is to extend their analysis to the setting in which relative smoothness holds only locally rather than globally, and to measure convergence using both the Bregman distance to the minimizer as well as the function value gap. The key step for proving Proposition~\ref{prop:primal-gradient-scheme-convergence-local} is to establish the following lemma.

\begin{lemma}\label{lem:relative-smoothness-iteration-decrease}
    In the setting of Proposition~\ref{prop:primal-gradient-scheme-convergence-local}, for each iteration $t$ and any $x\in\mathcal C$ we have
    \begin{align*}
    D_h(x,x^{(t+1)})\leq \left(1-\frac\mu L\right)D_h(x,x^{(t)})+\frac1L(f(x)-f(x^{(t+1)})).
\end{align*}
\end{lemma}

To prove \cref{lem:relative-smoothness-iteration-decrease} we use the well-known three-point property of Bregman divergences.
\begin{lemma}[Three-Point Property, \cite{tseng2008accelerated}]\label{lem:three-point-property}
Let $\varphi\colon\mathcal C\to\R$ be convex. Given $z\in\R^d$, let $z^+ := \arg\min_{x\in \mathcal C} \{ \varphi(x) + D_h(x,z) \}$. Then,
\[
\varphi(x) + D_h(x,z) \;\ge\; \varphi(z^+) + D_h(z^+,z) + D_h(x,z^+),
\quad \forall x \in \mathcal C.
\]
\end{lemma}

\begin{proof}[Proof of Lemma~\ref{lem:relative-smoothness-iteration-decrease}]
By the $L$-locally relative smoothness condition, for any iteration $t$, we have
\begin{align*}
    f(x^{(t+1)})
    \leq f(x^{(t)})+\langle\nabla f(x^{(t)}),x^{(t+1)}-x^{(t)}\rangle+LD_h(x^{(t+1)},x^{(t)}).
\end{align*}
Applying Lemma~\ref{lem:three-point-property} with $\varphi(x)\coloneqq\langle\nabla f(x^{(t)}),x-x^{(t)}\rangle$ and using the fact that $x^{(t+1)}=\argmin_{x\in\mathcal C}\big\{\varphi(x)$ $ +LD_h(x,x^{(t)})\big\}$, we obtain that for any $x\in\mathcal C$,
\begin{align*}
\langle\nabla f(x^{(t)}),x^{(t+1)}-x^{(t)}\rangle
&\leq \langle\nabla f(x^{(t)}),x-x^{(t)}\rangle+LD_h(x,x^{(t)})\\
&\qquad\qquad-LD_h(x^{(t+1)},x^{(t)})-LD_h(x,x^{(t+1)}).
\end{align*}
Therefore,
\begin{align}
    f(x^{(t+1)})
    &\leq f(x^{(t)})+\langle\nabla f(x^{(t)}),x-x^{(t)}\rangle+LD_h(x,x^{(t)})-LD_h(x,x^{(t+1)})\nonumber\\
    &\leq f(x)+(L-\mu) D_h(x,x^{(t)})-LD_h(x,x^{(t+1)})
    \label{eq:monotonicity-reason},
\end{align}
or equivalently,
\begin{align*} 
    D_h(x,x^{(t+1)})\leq \left(1-\frac\mu L\right)D_h(x,x^{(t)})+\frac1L(f(x)-f(x^{(t+1)})).
\end{align*}
\end{proof}

\begin{proof}[Proof of Proposition~\ref{prop:primal-gradient-scheme-convergence-local}.]
By Lemma~\ref{lem:relative-smoothness-iteration-decrease}, for each iteration $t$ and any $x\in\mathcal C$ we have
\begin{align*} 
    D_h(x,x^{(t+1)})\leq \left(1-\frac\mu L\right)D_h(x,x^{(t)})+\frac1L(f(x)-f(x^{(t+1)})).
\end{align*}
For iteration $t$ this yields 
\begin{align*}
    D_h(x,x^{(t)})\leq \left(1-\frac\mu L\right)^{t}D_h(x,x^{(0)})+\frac1L\sum_{k\in[t]}\left(1-\frac\mu L\right)^{t-k}(f(x)-f(x^{(k)})).
\end{align*}
Substituting $x=x^*$ into the inequality above gives
\begin{align*}
    D_h(x^*,x^{(t)})\leq \left(1-\frac\mu L\right)^{t}D_h(x^*,x^{(0)}),
\end{align*}
and
\begin{align*}
\sum_{k\in[t]}\left(1-\frac\mu L\right)^{t-k}(f(x^{(k)})-f(x^*))
&\leq L\left(1-\frac\mu L\right)^tD_h(x^*,x^{(0)})-D_h(x^*,x^{(t)})\\
&\leq L\left(1-\frac\mu L\right)^tD_h(x^*,x^{(0)}),
\end{align*}
where we used that $f(x^*)-f(x^{(k)})\leq 0$ for all $k$, and that $D_h(x^*,x^{(t)})\geq 0$ since $h$ is convex. Substituting $x=x^{(t)}$ in \eqref{eq:monotonicity-reason}, we obtain $f(x^{(t+1)})\leq f(x^{(t)})-LD_h(x,x^{(t+1)})\leq f(x^{(t)})$, which gives
\begin{align*}
    \sum_{k\in[t]}\left(1-\frac\mu L\right)^{t-k}(f(x^{(k)})-f(x^*))
    &\geq \sum_{k\in[t]}\left(1-\frac\mu L\right)^{t-k}(f(x^{(t)})-f(x^*))\\
    &=\frac{L}{\mu}\Big(1-\Big(1-\frac\mu L\Big)^t\Big)(f(x^{(t)})-f(x^*)),
\end{align*}
and therefore
\[
f(x^{(t)})-f(x^*)\leq \frac{\mu(1-\mu/L)^t}{1-(1-\mu/L)^t}\cdot D_h(x^*,x^{(0)}).
\]
\end{proof}

\section{Algorithm based on a matrix potential \texorpdfstring{$\fobj_\dual$}{Fmat}}\label{sec:matrix_function_definition}

In this section, we present an algorithm that computes $\epsilon$-estimates of Lewis weights by approximately solving \eqref{eq:mat-opt-problem} via the locally relatively smooth gradient descent framework in Section~\ref{sec:relative_smoothness_framework}. Throughout the section, we let $h_\dual(M)\coloneqq-\logdet(M)$ for all $M\in\mathcal S^n_{\succ 0}$, and denote $\bar\epsilon=\frac{\epsilon}{2(1+\alphap)}$.

\subsection{Properties of \texorpdfstring{$\fobj_\dual$}{Fmat}}

Here we present several properties of $\fobj_\dual$, including its duality with $\fobj_\vec$, explicit formulas for its gradient, Hessian and optimum, and its relative strong convexity and local smoothness properties with respect to $h_\dual$.

\begin{lemma}\label{lem:duality}
The optimization problems \eqref{eq:vec-opt-problem} and \eqref{eq:mat-opt-problem} are dual to each other in the following sense:
\begin{align*}
    \underset{M\succ 0}{\min}\,\fobj_\dual(M)
    =n-\underset{w>0}{\min}\,\fobj_\primal(w)
\end{align*}
\end{lemma}

\begin{proof}
Observe that
\begin{align}
\label{eq:minimax-form}
\underset{v>0}{\min}\,\fobj_\primal(v)
&=\underset{v>0}{\min}\,\underset{M\succ 0}{\max}\,
\bigg[\log\det(M)+n-\mathrm{Tr}(M A^\top V A)+\frac{1}{1+\alphap}\sum_{i\in[m]} v_i^{1+\alphap}\bigg]\,.
\end{align}
Define
\[
\Phi(v,M)\coloneqq\log\det(M)-\mathrm{Tr}(M A^\top V A)+\frac{1}{1+\alphap}\sum_{i\in[m]} v_i^{1+\alphap}.
\]
Then, 
$\Phi(v,M)$ is convex with respect to $v\in\R^m_{>0}$ and concave with respect to $M\in\R^{n\times n}_{\succ 0}$. Moreover, $\fobj_\primal$ diverges to $+\infty$ whenever any coordinate $v_i\to 0$ or $v_i\to\infty$. Therefore, $\fobj_\primal$ admits a finite minimizer, and both the minimization over $v$ and the maximization over $M$ in \eqref{eq:minimax-form} may be restricted to compact convex subsets without changing their values. Applying Sion’s minimax theorem on these restricted domains then yields
\begin{align*}
\underset{v>0}{\min}\,\fobj_\primal(v)
&=\underset{M\succ 0}{\max}\,\underset{v>0}{\min}\,\bigg[\log\det(M)+n-\mathrm{Tr}(M A^\top V A)+\frac{1}{1+\alphap}\sum_{i\in[m]} v_i^{1+\alphap}\bigg]
\end{align*}
Furthermore, for any $M\in\R^{n\times n}_{\succ 0}$,
\begin{align*}
\min_{v>0}\Phi(v,M)&=\log\det(M)+\sum_{i\in[m]} \inf_{v_i>0}\left(\frac{1}{1+\alphap}v_i^{1+\alphap}-v_i\, a_i^\top M a_i\right)\\
&=\log\det(M)-\frac2p\sum_{i\in[m]} (a_i^\top M a_i)^{p/2}=-\fobj_\dual(M),
\end{align*}
which gives
\begin{align*}
    \underset{M\succ 0}{\min}\,\fobj_\dual(M)
    =n-\underset{v>0}{\min}\,\fobj_\primal(v).
\end{align*}
\end{proof}

\begin{lemma}\label{lem:fobj_dual_properties}
For any $M \in \Sn_{\succ 0}$ the gradient and Hessian of $\fobj_\dual$ have the following expressions: 
\begin{align*}
\nabla \fobj_\dual(M ) 
&= -M ^{-1} + \sum_{i\in[m]} (a_i^\top M  a_i)^{\betap} a_i a_i^{\top}\text{ and }\\
\nabla^2 \fobj_\dual(M) 
&= M ^{-1} \otimes M ^{-1} + \betap \sum_{i\in[m]} (a_i^\top M  a_i)^{\betap-1} a_i a_i^\top \otimes a_i a_i^{\top}\,.
\nonumber
\end{align*}
Moreover, $\fobj_\dual$ has a unique minimizer  $M_{*}$ in $\Sn_{\succ 0}$ that satisfies $M_{*}^{-1} = A^\top  V_*A$.
\end{lemma}

\begin{proof}
For any $M\in\R^{d\times d}_{\succ 0}$ and $H\in\R^{d\times d}$, we have 
\[
\frac{\d}{\d t}\Big|_{t=0}\log\det(M+tH)=\operatorname{Tr}(M^{-1}H)=\langle M^{-1},H\rangle,
\]
and since $p/2 - 1 = \frac{p - 2}{2} = \betap$
\[
\frac{\d}{\d t}\Big|_{t=0}(a_i^\top(M+tH)a_i)^{p/2}
=\frac{p}{2}(a_i^\top M a_i)^{\betap}a_i^\top H a_i
\]
for any $i\in[m]$. Hence,
\[
\frac{\d \fobj_\dual(M+tH)}{\d t}\Big|_{t=0}
=\Big\langle -M^{-1}+\sum_{i\in[m]} (a_i^\top M a_i)^{\betap}a_i a_i^\top,H\Big\rangle,
\]
which gives
\[
\nabla \fobj_\dual(M)
=-M^{-1}+\sum_{i\in[m]} (a_i^\top M a_i)^{\betap}a_i a_i^\top.
\]
Similarly, for any $M\in\R^{d\times d}_{\succ 0}$ and $H\in\R^{d\times d}$ we have
\[
\lim_{t\to 0}\frac{\nabla\fobj_\dual(M+tK)-\nabla\fobj_\dual(M)}{t}
= M^{-1}KM^{-1}
+\frac1\alphap\sum_{i\in[m]} (a_i^\top M a_i)^{\betap-1}(a_i^\top K a_i)\,a_i a_i^\top,
\]
which gives
\[
\frac{\d}{\d t}\Big|_{t=0}
\big\langle \nabla F(M+tK), H \big\rangle=\big\langle M^{-1}KM^{-1},H\big\rangle
+\frac1\alphap\sum_{i\in[m]} (a_i^\top M a_i)^{\betap-1}
\langle a_i a_i^\top,H\rangle \langle a_i a_i^\top,K\rangle,
\]
for all $M\in\R^{d\times d}_{\succ 0}$ and $H,K\in\R^{d\times d}$, which implies
\[
\nabla^2 \fobj_\dual(M)
= M^{-1} \otimes M^{-1}
+\frac1\alphap\sum_{i\in[m]} (a_i^\top M a_i)^{\betap-1}
a_i a_i^\top \otimes a_i a_i^\top.
\]
Note that $\nabla^2 \fobj_\dual(M)\succ 0$. Hence, $\fobj_\dual$ is strictly convex and has a unique minimizer $M_{*}$ satisfying $\nabla \fobj_\dual(M_{*}) =0$, or equivalently,
\[
M_{*} ^{-1} = \sum_{i\in[m]} (a_i^\top M_{*}  a_i)^{\betap} a_i a_i^{\top}.
\]
Let $u\in\R^m_{>0}$ be the vector with coordinates $[u_*]_i=(a_i^\top M_{*}  a_i)^{\betap}$. Then, we have $M_{*}=\big(AUA\big)^{-1}$ and
\[
[u]_i=\big(a_i^\top \big(AUA\big)^{-1}a_i\big)^{\betap},
\]
showing that $u=v_*$.
\end{proof}

\begin{lemma}\label{lem:dual_relative_convex}
$\fobj_\dual$ is $1$-relatively strongly convex with respect to $h_\dual$. Moreover, if $M=(A^\top VA)^{-1}$ for some $V=\mDiag(v)$ with $v\in\R^m_{>0}$, we have
\begin{equation}
\label{eq:dual_smoothness_local}
\nabla^2\fobj_\dual(M ) \preceq \left(1+\betap \Phi_\max(v)\right) \nabla^2 h_\dual(M).
\end{equation}
\end{lemma}
\begin{proof}
The $1$-relative strong convexity of $\fobj$ with respect to $h_{\dual}$ follows from the convexity of $(a_i^\top M a_i)^{1 + \betap}= (a_i^\top M  a_i)^{p/2}$. More formally, 
observe that 
\[
\nabla^2 \fobj_\dual(M)
= \nabla^2h_\dual(M)
+ \frac1\alphap\sum_{i\in[m]} (a_i^\top M a_i)^{\betap-1}
a_i a_i^\top \otimes a_i a_i^\top\succeq \nabla^2h_\dual(M),
\]
proving that $\fobj_\dual$ is $1$-relatively strongly convex with respect to $h_\dual(M)$.

We now establish \eqref{eq:dual_smoothness_local}. First, using $M^{1/2} a_i a_i^\top M^{1/2} \preceq a_i^\top M a_i I$ and monotonicity of the Kronecker product, we observe that 
\begin{align*}
\sum_{i\in[m]} (a_i^\top M  a_i)^{\betap-1} M ^{1/2} a_i a_i^\top M ^{1/2} \otimes M ^{1/2} a_i a_i^\top M ^{1/2} \preceq M ^{1/2} \left( \sum_{i\in[m]} (a_i^\top M  a_i)^{\betap} a_i a_i^{\top}\right) M ^{1/2} \otimes I. \end{align*}
If $M=(A^\top VA)^{-1}$ for some $V=\mDiag(v)$, we have
\begin{align}
\label{eq:matrix-product-to-the-power-of-beta}
(a_i^\top M  a_i)^{\betap}=(a_i^\top (A^\top VA)^{-1}  a_i)^{\betap}=\Phi_i(v)v_i
\end{align}
and therefore
\begin{align*}
&\sum_{i\in[m]} (a_i^\top M  a_i)^{\betap-1} M ^{1/2} a_i a_i^\top M ^{1/2} \otimes M ^{1/2} a_i a_i^\top M ^{1/2} \\
&\quad\qquad\preceq M ^{1/2} \left( \sum_{i\in[m]} v_i \Phi_i(v) a_i a_i^\top \right) M ^{1/2}  \otimes I \nonumber\\
&\quad\qquad\preceq \Phi_\max(v) M ^{1/2} \left(\sum_{i\in[m]} v_ia_i a_i^\top \right) M ^{1/2}  \otimes I \nonumber\\
&\quad\qquad= \Phi_\max(v) I\otimes I\,.
\end{align*}
Equation \eqref{eq:dual_smoothness_local} then follows as $\nabla^2(-\logdet(M)) = M^{-1} \otimes M^{-1}$ and therefore 
\[
\sum_{i\in[m]} (a_i^\top M a_i)^{\betap-1}a_i a_i^\top \otimes a_i a_i^\top\preceq\Phi_\max(v) \nabla^2 h_\dual(M).
\]
\end{proof}

\begin{lemma} 
$D_{h_\dual}\big((A^\top  V_{*}A)^{-1},(A^\top A)^{-1}\big)\leq m-n.$
\end{lemma}
\begin{proof}
    By the definition of the Bregman divergence, we have
    \begin{align*}
        &D_{h_\dual}\big((A^\top  V_{*}A)^{-1},(A^\top A)^{-1}\big)\\
        &\qquad=-\logdet((A^\top  V_{*}A)^{-1})+\logdet ((A^\top A)^{-1})+\big\langle A^\top A,(A^\top  V_{*}A)^{-1}-(A^\top A)^{-1}\big\rangle.
    \end{align*}
    Using $A^\top V_* A \preceq A^\top A$, we have
    \begin{align*}
        -\logdet((A^\top  V_{*}A)^{-1})+\logdet ((A^\top A)^{-1})=\logdet(A^\top  V_{*}A)-\logdet (A^\top A)\leq 0.
    \end{align*}
    Moreover, since $[v_{*}]_i= [\sigma_p(A)]_i^{\frac{\betap}{1+\betap}}\leq 1$ for any $i\in[m]$, we have
    \begin{align*}
        \big\langle A^\top A,(A^\top  V_{*}A)^{-1}-(A^\top A)^{-1}\big\rangle
        =\Tr[A(A^\top  V_{*}A)^{-1}A^\top ]-n=\sum_{i\in[m]}[v_*]_i ^{1/\betap}-n\leq m-n.
    \end{align*}
    Hence, we can conclude that $D_{h_\dual}\big((A^\top  V_{*}A)^{-1},(A^\top A)^{-1}\big)\leq m-n$.
\end{proof}

\subsection{Applying the local relative smoothness framework to \texorpdfstring{$\fobj_\dual$}{Fmat}}\label{sec:applying_rs_to_dual}

Here we give our algorithm that solves \eqref{eq:mat-opt-problem} by iteratively performing the update in \eqref{eq:algintro}.  
We show that in every iteration of the algorithm, the iterate $M^{(t)}$ obtained from the primal gradient scheme is of the form $M^{(t)} = (A^\top V^{(t)}A)^{-1}$, where $V^{(t)} = \mDiag(v^{(t)})$ for $v^{(t)} \in \R^m_{>0}$, see \cref{lem:dual_algo_relative_smoothness_interpretation}. The algorithm and its analysis are therefore stated in terms of the sequence of vectors $v^{(t)}$. 
We then establish in \Cref{lem:dual_algorithm_rhomax_control} that $\Phi_{\max}(v^{(t)})$ remains bounded for any iteration, and use this in \Cref{lem:relative-smoothnes-between-iterations-dual} to verify local relative smoothness between consecutive iterates. The local relative smoothness framework then allows us to establish~\Cref{thm:relative-smoothness-dual-main}.

\begin{algorithm}[htbp!]
\DontPrintSemicolon
\SetAlgoLined
\caption{High-precision algorithm using the matrix potential $\fobj_\dual$}\label{alg:relative-smoothness-dual}
\KwIn{non-degenerate $A \in \mathbb{R}^{m \times n}$, $p> 2$, $\epsilon>0$
}
Set $L=32p\barbetap$, $\totaliters=\lceil 4L\log(2m/\bar\epsilon)\rceil$ where $\bar \epsilon=\frac{\epsilon}{2(1+\alphap)}$, and $v^{(0)}_i=1$ for all $i\in[m]$.

\For{$t= 0, 1, \dotsc, \totaliters-1$}{
$v_i^{(t+1)}
\leftarrow\Big(1+ \frac{\Phi_i(v^{(t)})-1}{L} \Big)v_i^{(t)}$, $\forall i\in[m]$\label{lin:fobj-rs-update-dual}\qquad \tcp*{Recall $\Phi_i(v) = \frac{(a_i^\top (A^\top V A)^{-1} a_i)^{\betap}}{v_i}$.}
}
Return $\widehat w\in\R^m_{>0}$, where $\widehat w_i=(a_i^\top(A^\top V^{(\totaliters)}A)^{-1}a_i)^{1 + \betap}$.
\end{algorithm}

\theomat* 

\begin{lemma}\label{lem:dual_algo_relative_smoothness_interpretation}
For any iteration $t$ in Algorithm~\ref{alg:relative-smoothness-dual}, we have\footnote{Since $\fobj_\dual$ is convex and $\lim_{M\to\partial\{M\mid M\succeq 0\}}\fobj_\dual(M)=+\infty$, the minimizer is always attained in the interior. 
}
\begin{align*}
M ^{(t+1)}=\argmin_{M  \succeq 0}\left\{\fobj_\dual(M ^{(t)})+\langle \nabla \fobj_\dual(M ^{(t)}),M -M ^{(t)}\rangle+L D_{h_\dual}(M ,M ^{(t)})\right\}
\end{align*}
\end{lemma}
\begin{proof}
Given the choice of $h_\dual$, we have that
\begin{align*}
D_{h_\dual}(M,M^{(t)}) = -\logdet(M) + \logdet(M^{(t)}) + \langle [M^{(t)}]^{-1}, M-M^{(t)}\rangle
\end{align*}
and, by Lemma~\ref{lem:fobj_dual_properties}, stationarity is
\begin{align}\label{eq:dual_gradient_zero_condition}
\begin{aligned}
&0 = \nabla \fobj_\dual(M ^{(t)})+L\nabla D_{h_\dual}(M,M^{(t)})\big|_{M=M^{(t+1)}}\\
&\qquad=\nabla \fobj_\dual(M ^{(t)}) + L\big(-[M^{(t+1)}]^{-1}+[M^{(t)}]^{-1}\big)\\
&\qquad=(L-1)[M^{(t)}]^{-1}+ \sum_{i\in[m]} (a_i^\top M^{(t)} a_i)^{\betap} a_i a_i^\top-L[M^{(t+1)}]^{-1}\,.
\end{aligned}
\end{align}
Substituting $M^{(t)}= (A^\top V^{(t)}A)^{-1}$ for any $t$, we have $(a_i^\top M^{(t)}  a_i)^{\betap}=\Phi_i(v^{(t)})v_i^{(t)}$ by \eqref{eq:matrix-product-to-the-power-of-beta}. Then, \eqref{eq:dual_gradient_zero_condition} equals
\begin{align*}
    \sum_{i\in[m]}\left((L+\Phi_i(v^{(t)})-1)v^{(t)}_i-Lv^{(t+1)}_i\right)a_ia_i^\top=0, 
\end{align*}
which is satisfied when $v_i^{(t+1)}
 = \Big(1+ \frac{\Phi_i(v^{(t)})-1}{L} \Big)v_i^{(t)}$ for all $i\in[m]$. 
Since $D_{h_\dual}(M,M^{(t)})$ is convex, we can conclude that
\begin{align*}
M ^{(t+1)}=\argmin_{M  \succeq 0}\left\{\fobj_\dual(M ^{(t)})+\langle \nabla \fobj_\dual(M ^{(t)}),M -M ^{(t)}\rangle+L D_{h_\dual}(M ,M ^{(t)})\right\}.
\end{align*}
\end{proof}

\begin{lemma}\label{lem:dual_algorithm_rhomax_control}
    For any iteration $t$ in Algorithm~\ref{alg:relative-smoothness-dual}, if $\Phi_\max(v^{(t)}) \leq 4\bar{\betap} \leq L$ then $\Phi_\max(v^{(t+1)})\leq 4\bar{\betap}$. 
\end{lemma}

\begin{proof}
    By the update formula in Line~\ref{lin:fobj-rs-update-dual}, we have $v_i^{(t+1)} \geq \big(1-\frac{1}{L}\big) v_i^{(t)}$, which leads to
    \begin{align*}
    \Phi_i(v^{(t+1)}) &= \frac{v_i^{(t)}}{v_i^{(t+1)}}\cdot \frac{(a_i^\top(A^\top V^{(t+1)}A)^{-1} a_i)^{\betap}}{v_i^{(t)}}\\
    &\leq \left(1-\frac{1}{L}\right)^{-\betap}\frac{v_i^{(t)}}{v_i^{(t+1)}} \cdot \frac{(a_i^{\top}(A^\top V^{(t)}A)^{-1} a_i)^{\betap}}{v_i^{(t)}} \\
    &= \left(1-\frac{1}{L}\right)^{-\betap}\left(1+\frac{\Phi_i(v^{(t)})-1}{L}\right)^{-1} \cdot \Phi_i(v^{(t)})\,.
    \end{align*}
    Since $L>1$, the function $\psi\colon\R^+\to\R$ defined as $\psi(x) \defeq \left(1-\frac{1}{L}\right)^{-\betap}\left(1+\frac{x-1}{L}\right)^{-1}x$, is monotonically increasing for $x>0$. Then, using $\Phi_{\max}(v^{(t)})\leq 4\barbetap$, we have
    \begin{align*}
    \Phi_\max(v^{(t+1)})
    \leq \left(1-\frac{1}{L}\right)^{-\betap}\left(1+\frac{4\barbetap-1}{L}\right)^{-1} 4\barbetap
    \leq \left(1-\frac{\betap}{ L}\right)^{-1}\left(1+\frac{4\barbetap}{2L}\right)^{-1}4\barbetap\leq4\barbetap,
    \end{align*}
    where the second inequality uses $L\ge 4\barbetap$ and $4\barbetap\ge 2$, and the fact that $(1-x)^{-\betap}\le (1-\betap x)^{-1}$ for all $0\le x\le 1/(2\betap)$. The third inequality uses $\frac{\betap}{L}\le \frac14$, and the fact that $(1-x)^{-1}(1+2x)^{-1}\le 1$ for all $0\le x\le 1/4$. 
\end{proof}

\begin{lemma}\label{lem:relative-smoothnes-between-iterations-dual}
For any iteration $t$ in Algorithm~\ref{alg:relative-smoothness-dual} and any $\lambda\in[0,1]$, if $\Phi_\max(v^{(t)}) \leq 4\barbetap$ and $L\geq 16p\barbetap$, 
\begin{align*}
\nabla^2\fobj_\dual(M_\lambda)\preceq \big(1+16p\barbetap\big)\nabla^2h_\dual(M_\lambda)
    \enspace\text{ where }
    \enspace
    M_\lambda \defeq (1-\lambda)M^{(t)}+\lambda M^{(t+1)}
    \,.
\end{align*}
\end{lemma}
\begin{proof}
By the update formula in Line~\ref{lin:fobj-rs-update-dual}, we have $\frac{|v_i^{(t+1)}-v_i^{(t)}|}{v_i^{(t)}}\leq \frac{|\Phi_i(v^{(t)})-1|}{L}\leq\frac{4\betap-1}{L}\leq \frac14$ where the second inequality follows from $\Phi_{\max}(v^{(t)})\le4\barbetap$, and the last inequality uses $L\ge 16p\barbetap\geq 16\barbetap$. Consequently, $\frac 34 A^\top V^{(t)}A\preceq A^\top V^{(t+1)}A\preceq \frac54A^\top V^{(t)}A$, and thus $\frac 45 M^{(t)}\preceq M_\lambda\preceq \frac43 M^{(t)}$ for any $\lambda \in[0,1]$. By Lemma~\ref{lem:fobj_dual_properties}, the Hessian of $\fobj_\dual$ admits the decomposition
\[
\nabla^2 \fobj_\dual(M_\lambda)
= M_\lambda^{-1} \otimes M^{-1}_\lambda
+\frac1\alphap\sum_{i\in[m]} (a_i^\top M_\lambda a_i)^{\betap-1}
a_i a_i^\top \otimes a_i a_i^\top.
\]
We bound the second term as follows:
\begin{align}
&\sum_{i\in[m]} (a_i^\top M_{\lambda} a_i)^{\betap-1} M_{\lambda} ^{1/2} a_i a_i^\top M_{\lambda}^{1/2} \otimes M_{\lambda}^{1/2} a_i a_i^{\top} M_{\lambda}^{1/2}\nonumber \\
&\qquad\preceq2\sum_{i\in[m]} (a_i^\top M_{\lambda} a_i)^{\betap-1} [M^{(t)}]^{1/2} a_i a_i^\top [M^{(t)}]^{1/2} \otimes [M^{(t)}]^{1/2} a_i a_i^{\top} [M^{(t)}]^{1/2} \nonumber\\
&\qquad\preceq2[M^{(t)}]^{1/2} \left( \sum_{i\in[m]} (a_i^\top M_{\lambda}  a_i)^{\betap} a_i a_i^{\top}\right) [M^{(t)}]^{1/2} \otimes I.
\label{eq:Mlambda_Hessian_mid_decomposition}
\end{align}
where the first inequality uses the spectral closeness between $M_\lambda$ and $M^{(t)}$. Next, by Lemma~\ref{lem:dual_algorithm_rhomax_control}, for each coordinate $i$ we have
\begin{align*}
    (a_i^\top M_{\lambda}  a_i)^{\betap}
    &\leq (\max\{a_i^\top M^{(t)}a_i,a_i^\top M^{(t+1)}a_i\})^{\betap}\\
    &\leq \max\big\{v_i^{(t)}\Phi_i(v^{(t)}),v_i^{(t+1)}\Phi_i(v^{(t+1)})\big\}\leq 8\barbetap v_i^{(t)}.
\end{align*}
Therefore,
\begin{align*}
    [M^{(t)}]^{1/2} \bigg( \sum_{i\in[m]} (a_i^\top M_{\lambda}  a_i)^{\betap} a_i a_i^\top\bigg) [M^{(t)}]^{1/2} 
    \preceq 8\barbetap[M^{(t)}]^{1/2} \bigg( \sum_{i\in[m]} v_i^{(t)} a_i a_i^{\top}\bigg) [M^{(t)}]^{1/2} 
    =8\barbetap I,
\end{align*}
which combined with Eq.~\eqref{eq:Mlambda_Hessian_mid_decomposition} gives
\[
\sum_{i\in[m]} (a_i^\top M_{\lambda} a_i)^{\betap-1} M_{\lambda} ^{1/2} a_i a_i^\top M_{\lambda}^{1/2} \otimes M_{\lambda}^{1/2} a_i a_i^{\top} M_{\lambda}^{1/2} \preceq 16\barbetap I\otimes I.
\]
Consequently,
$\nabla^2\fobj_\dual(M_\lambda)\preceq \left(1+16p\barbetap\right)\nabla^2h_\dual(M_{\lambda})$.
\end{proof}

\begin{proof}[Proof of Theorem~\ref{thm:relative-smoothness-dual-main}.]
By Lemma~\ref{lem:dual_algo_relative_smoothness_interpretation}, each iteration of Algorithm~\ref{alg:relative-smoothness-dual} can be equivalently written as
\begin{align*}
M ^{(t+1)}=\argmin_{M  \succ 0}\left\{\fobj_\dual(M ^{(t)})+\langle \nabla \fobj_\dual(M ^{(t)}),M -M ^{(t)}\rangle+L D_{h_\dual}(M ,M ^{(t)})\right\},
\end{align*}
where $M^{(t)}=(A^\top V^{(t)}A)^{-1}$. Since 
\[
\Phi_\max(v^{(0)})=\bigg(\max_{i\in[m]}\frac{\sigma_i(v^{(0)})}{[v_i^{(0)}]^{1+1/\betap}}\bigg)^{\betap}=\max_{i\in[m]}\sigma_i(v^{(0)})^{\betap}\leq 1,
\]
Lemma~\ref{lem:dual_algorithm_rhomax_control} implies that $\Phi_\max(v^{(t)})\leq 4\barbetap$ for all iterations $t$. By Lemma~\ref{lem:relative-smoothnes-between-iterations-dual}, it follows that $\fobj_\dual$ is $(1+16p\barbetap)\leq L$-smooth relative to $h_\dual$ between any two consecutive iterates $M^{(t)}$ and $M^{(t+1)}$. Moreover, since $\fobj_\dual$ is 1-strongly convex relative to $h_\dual$ by Lemma~\ref{lem:dual_relative_convex}, Proposition~\ref{prop:primal-gradient-scheme-convergence-local} yields 
\begin{align*}
    D_{h_\dual}(M_{*},M^{(\totaliters)})
    &\leq \left(1-\frac1 L\right)^{\totaliters}D_{h_\dual}(M_{*},M^{(0)})+\frac1L\sum_{t\in[\totaliters]}\left(1-\frac1 L\right)^{\totaliters-t}(\fobj_\dual(M_{*})-\fobj_\dual(M^{(t)}))\\
    &\leq \left(1-\frac1 L\right)^{\totaliters}D_{h_\dual}(M_{*},M^{(0)})\leq\frac{\bar\epsilon^2}{16} 
\end{align*}
where the second inequality uses that $M_{*}=\arg\underset{M\succ 0}{\min}\,\fobj_\dual(M)$.  Lemma~\ref{lem:bregman_divergence_to_spectral_difference} then implies $M^{(\totaliters)}\approx_{\bar\epsilon/2} M_{*}$. Consequently, $\widehat v \in \R^{m}_{>0}$ with $\widehat v_i =(a_i^\top(A^\top V^{(T)}A)^{-1}a_i)^{\betap}$ for all $i \in [m]$ satisfies
\[
    \left|\frac{\widehat v_i}{[v_*]_i}-1\right|=\left|\frac{a_i^\top M^{(\totaliters)}a_i}{a_i^\top M_{*}a_i}-1\right|\leq \frac{\bar\epsilon} 2\text{ for all }i\in [m]\,.
\]
Therefore $(1-\epsilon )\sigma_p(A)\leq \widehat w\leq (1+\epsilon)\sigma_p(A)$ since $\sigma_p(A)=v_{*}^{1+1/\betap}$ and $\widehat w=\widehat v^{1+1/\betap}$.
\end{proof}

\begin{remark}\label{rem:dual-coordinatewise-scaling-change}
By Lemma~\ref{lem:dual_algorithm_rhomax_control} and the initialization, $\Phi_{\max}(v^{(t)})\le 4\barbetap$ for every iteration $t$ in Algorithm~\ref{alg:relative-smoothness-dual}. Therefore, for every coordinate $i\in[m]$,
\[
\left|\frac{v_i^{(t+1)}}{v_i^{(t)}}-1\right|
=
\frac{|\Phi_i(v^{(t)})-1|}{L}
\le
\frac{4\barbetap}{32p\barbetap}
=
\frac{1}{8p}
\le \frac14 .
\]
Consequently, the diagonal scaling $D^{(t)}=(V^{(t)})^{1/2}$ used in each leverage score computation satisfies
\[
\sqrt{\frac34}\,D_{ii}^{(t)}
\le
D_{ii}^{(t+1)}
\le
\sqrt{\frac54}\,D_{ii}^{(t)},
\quad \forall i\in[m].
\]
\end{remark}

\section{A relative smoothness algorithm using \texorpdfstring{$\fobj_\primal$}{Fvec}}\label{sec:applying_rs_to_primal}

In this section, we present an algorithm that computes $\epsilon$-estimates of Lewis weights by approximately solving \eqref{eq:vec-opt-problem} via the locally relatively smooth gradient descent framework in Section~\ref{sec:relative_smoothness_framework}. Algorithm~\ref{alg:relative-smoothness} has a slightly worse logarithmic dependence than Algorithm~\ref{alg:relative-smoothness-dual}, but it demonstrates the broader applicability of our relative smoothness framework and gives a streamlined analysis of the potential function studied in \cite{fazel22}. 
Throughout this section, we set $\overline \rho = 4\barbetap$, and define $h_{\primal}(v)\coloneqq\frac{1}{1 + \alphap} \sum_{i\in[m]}v_i^{1 + \alphap}$.

\begin{algorithm}[htbp!]
\DontPrintSemicolon
\caption{High-precision algorithm using $\fobj_\primal$}\label{alg:relative-smoothness}
\KwIn{non-degenerate $A \in \mathbb{R}^{m \times n}$, $p>2$, accuracy $\epsilon \in (0, \frac14]$
}

Set $\hat\epsilon=\frac{\epsilon}{\barbetap n}\min\left\{\frac{1}{96(p-2)^2(4p-7)^2},\frac{1}{50}\right\}$, $\Delta=\min\left\{\frac{\hat\epsilon^3n}{384\baralphap},\frac{\hat\epsilon^2\alphap^3}{27\times 10^3}\right\}$, $L=32p\barbetap$, $\totaliters=2L\max\left\{\ln\big(\frac{4m}{\Delta}\big),4\ln\big(\frac5{\alphap\hat\epsilon}\big)\right\}$, and $v^{(0)}_i=1$ for all $i\in[m]$. \;

\For{$t= 0, 1, \dotsc, \totaliters-1$}{
$v_i^{(t+1)}\leftarrow\Big(1+\frac{\rho_i(v^{(t)})-1}{L}\Big)^{1/\alphap}v_i^{(t)},\quad\forall i\in[m]$\label{lin:fobj-rs-update}\;
}	

Return $\hw \in \R^m_{>0}$, where $\hw_i = (a_i^\top ( A^\top  V^{(\totaliters)}  A)^{-1} a_i)^{1+1/\alphap}$. 
\end{algorithm}

\theovec*

To prove Theorem~\ref{thm:primal-algorithm-guarantee}, we first present several properties of $\fobj_\primal$ in Section~\ref{sec:fobj_primal_properties}. We then show in Section~\ref{sec:primal_alg_func_value_decrease} that the function value gap becomes sufficiently small after half of the iterations. Finally, in Section~\ref{sec:primal_alg_one_sidedness_last_iterate}, we establish that the final iterate $v^{(T)}$ of Algorithm~\ref{alg:relative-smoothness} gives a one-sided $\hat\epsilon$-approximate Lewis weight vector, which can be converted into an $\epsilon$-estimate using Theorem~\ref{thm:one-sided-to-multiplicative}. We present \Cref{alg:relative-smoothness} with the parameter $\eps$ as input since the main purpose of \Cref{alg:relative-smoothness} is to compute $\eps$-estimates of $\sigma_p(A)$, but we point out that the one-sided $\hat \eps$-approximation of \Cref{thm:primal_algorithm_one_sided_guarantee} holds for any $\hat \eps \in (0,\frac{1}{4}]$.

\subsection{Properties of \texorpdfstring{$\fobj_\primal$}{Fvec}}\label{sec:fobj_primal_properties}
Here we present several properties of $\fobj_\primal$, including explicit formulas for its gradient and Hessian, an upper bound on its function value gap for any $v\in\R^{m}_{\geq 0}$, and the fact that it is convex and locally smooth relative to $h_{\primal}$. We will use the projection matrix $P(v)\coloneqq V^{1/2}A(A^\top VA)^{-1}A^\top V^{1/2}$, and write $P(v)^{(2)}$ for the Schur product (entry-wise product) of $P(v)$ with itself. 

\begin{lemma}[Gradient and Hessian, Lemma 3 of \cite{fazel22}]\label{lem_gradAndHess}
For any $v\in \R^m_{>0}$, the gradient and Hessian of $\fobj_\primal$ have the following expressions:
$$
\big[ \nabla \fobj_\primal(v)\big]_i = v_i^{-1} \cdot (v_i^{1+\alphap} - \sigma_i(v)),\qquad\nabla^2 \fobj_\primal(v) =  V^{-1} P(v)^{(2)} V^{-1} + \alphap  V^{\alphap-1}.
$$
\end{lemma}

\begin{lemma}[Lemma 6 of \cite{fazel22}]\label{lem:fobj_primal_gap_lower}
For any $v \in \R^m_{> 0}$, we have 
\begin{align*}
\fobj_\primal(v) -\fobj_\primal(v_{*})  \geq \frac{1}{6\baralphap}\sum_{i \in [m]}v_i^{1 + \alphap} \cdot \frac{(\rho_i(v) - 1)^2}{\rho_i(v) + 1}.
\end{align*}
\end{lemma}

\begin{lemma}\label{lem:fobj-relative-smoothness}
$\fobj_\primal$ is 1-strongly convex relative to $h_{\primal}(v)\coloneqq\frac{1}{1 + \alphap} \sum_{i\in[m]}v_i^{1 + \alphap}$. Moreover,
\[
\nabla^2 \fobj_\primal(v) \preceq \left(1 + \frac{\rho_{\max}(v)}{\alphap} \right) \nabla^2 h_{\primal}(v),\quad \forall v \in \R^m_{>0}.
\]
\end{lemma}
\begin{proof}
The gradient and Hessian of $h_{\primal}$ satisfy
\begin{align*}
[\nabla h_{\primal}(v)]_i=v_i^{\alphap},\qquad\nabla^2 h_{\primal}(v)=\alphap\cdot V^{\alphap-1}.
\end{align*}
Note that 
\[
\mZero \preceq
\nabla^2 \left(-\logdet( A^\top  V  A)\right) 
    =  V^{-1} \mP(v)^{(2)}  V^{-1}
    \preceq   V^{-1} \mSigma  V^{-1}
    = \mDiag( V^{\alphap - 1} \rho),
\]
where we denote $\mSigma\coloneqq\mDiag(\sigma(v))$. Then by Lemma~\ref{lem_gradAndHess}, we can conclude that
\[
\nabla^2 h_{\primal}(v) \preceq \nabla^2 \fobj_\primal(v) \preceq \left(1 + \frac{\rho_{\max}(v)}{\alphap} \right) \nabla^2 h_{\primal}(v)\,,
\]
which shows that $\fobj_{\primal}$ is $1$-strongly convex relative to $h_{\primal}$.
\end{proof}

\subsection{Function value decrease in Algorithm~\ref{alg:relative-smoothness}}\label{sec:primal_alg_func_value_decrease}
Here we show that the value of $\fobj_\primal(v^{(t)})-\fobj_\primal(v_{*})$ is at most $\Delta$ after $t\geq \totaliters/2$ iterations in Algorithm~\ref{alg:relative-smoothness}.

\begin{lemma}
For any iteration $t$ of Algorithm~\ref{alg:relative-smoothness}, we have
\begin{align*}
v^{(t+1)}=\argmin_{v\in\mathbb{R}^m_{> 0}}\big\{\fobj_\primal(v^{(t)})+\big\langle\nabla \fobj_\primal(v^{(t)}),v-v^{(t)}\big\rangle+LD_{h_{\primal}}(v,v^{(t)})\big\}.
\end{align*}
\end{lemma}
\begin{proof}
Since the function
\[
\fobj_\primal(v^{(t)})+\langle\nabla \fobj_\primal(v^{(t)}),v-v^{(t)}\rangle+LD_{h_{\primal}}(v,v^{(t)})
\]
is convex, it has one unique minimizer. Given that
\begin{align*}
&\nabla\left(\fobj_\primal(v^{(t)})+\big\langle\nabla \fobj_\primal(v^{(t)}),v-v^{(t)}\big\rangle+LD_{h_{\primal}}(v,v^{(t)})\right)\Big|_{v=v^{(t+1)}}\\
&\qquad\quad= L\nabla h_{\primal}(v^{(t+1)})+\nabla\fobj_\primal(v^{(t)})-L\nabla h_{\primal}(v^{(t)})=0,
\end{align*}
we can conclude that
\[
v^{(t+1)}=\argmin_{v\in\mathbb{R}^m_{>0}}\big\{\fobj_\primal(v^{(t)})+\big\langle\nabla \fobj_\primal(v^{(t)}),v-v^{(t)}\big\rangle+LD_{h_{\primal}}(v,v^{(t)})\big\}.
\]
\end{proof}

The following two lemmas establish that $\fobj_\primal$ is relatively smooth with respect to $h_{\primal}$ between each pair of consecutive iterates $v^{(t)}$ and $v^{(t+1)}$ (and are analogs of Lemma~\ref{lem:dual_algorithm_rhomax_control} and \ref{lem:relative-smoothnes-between-iterations-dual} for $\fobj_\primal$).

\begin{lemma}\label{lem:algorithm_rhomax_control}
    Let $L \geq 4\barbetap$. For any iteration $t$ in Algorithm~\ref{alg:relative-smoothness}, if $\rho_\max(v^{(t)}) \leq 4\barbetap$, then $\rho_\max(v^{(t+1)}) \leq 4\barbetap$. Consequently, $\rho_\max(v^{(t)}) \leq 4\barbetap$ for all iterations $t$.
\end{lemma}
\begin{proof}
    The proof strategy is similar to Lemma~\ref{lem:dual_algorithm_rhomax_control}. First, note that $\rho_{\max}(v^{(0)}) \le 1 \le 4\barbetap$. We then show that for any $t$ satisfying $\rho_{\max}(v^{(t)}) \le 4\barbetap$, it also holds that $\rho_{\max}(v^{(t+1)}) \le 4\barbetap$. In particular, by the update formula in Line~\ref{lin:fobj-rs-update}, $v_i^{(t+1)} \geq \big(1-\frac{1}{L}\big)^{1/\alphap} v_i^{(t)}=\big(1-\frac{1}{L}\big)^{\betap} v_i^{(t)}$, which leads to
    \begin{align*}
    \rho_i(v^{(t+1)}) &= \left(\frac{v_i^{(t)}}{v_i^{(t+1)}}\right)^{\alphap}\cdot \frac{a_i^\top(A^\top V^{(t+1)}A)^{-1} a_i}{(v_i^{(t)})^{\alphap}}\nonumber \\
    &\leq \left(1-\frac{1}{L}\right)^{-\betap}\left(\frac{v_i^{(t)}}{v_i^{(t+1)}}\right)^{\alphap} \cdot \frac{a_i^\top(A^\top V^{(t)}A)^{-1} a_i}{(v_i^{(t)})^{\alphap}}\nonumber \\
    &= \left(1-\frac{1}{L}\right)^{-\betap}\left(1+\frac{\rho_i(v^{(t)})-1}{L}\right)^{-1} \cdot \rho_i(v^{(t)}),
    \end{align*}
    and the remainder of the proof proceeds analogously to proof of Lemma~\ref{lem:dual_algorithm_rhomax_control}.
    \end{proof}
    \begin{lemma}\label{lem:algorithm_rhomax_control_lambda}
For any iteration $t$ in Algorithm~\ref{alg:relative-smoothness} and any $\lambda \in [0,1]$, if $\rho_\max(v^{(t)}) \leq 4\barbetap$ and $L \geq 4\barbetap$,
    \[
    \rho_\max((1-\lambda) v^{(t)} + \lambda v^{(t+1)}) \leq 16\barbetap. 
    \]
\end{lemma}
\begin{proof}
For any $\lambda \in [0,1]$, let $u = (1-\lambda) v^{(t)} + \lambda v^{(t+1)}$. By the update formula in Line~\ref{lin:fobj-rs-update}, we have $u_i = \left(1+ \lambda\left(\big(1+\frac{\rho_i(v^{(t)})-1}{L}\big)^{1/\alphap}-1\right)\right) v_i^{(t)}$ for all $i \in [m]$, and therefore 
\[
u_i \geq \left(1+ \lambda\left(\left(1-\frac{1}{L}\right)^{1/\alphap}-1\right)\right) v_i^{(t)}.
\]
Since $\betap=1/\alphap$, we have
\begin{align}
    \rho_i(u) &= \left(\frac{v_i^{(t)}}{u_i}\right)^{\alphap}\cdot \frac{a_i^\top(A^\top UA)^{-1} a_i}{(v_i^{(t)})^{\alphap}} \notag\\
    &\leq\left(1+ \lambda\left(\left(1-\frac{1}{L}\right)^{\betap}-1\right)\right)^{-1}\left(\frac{v_i^{(t)}}{u_i}\right)^{\alphap} \cdot \frac{a_i^\top(A^\top V^{(t)}A)^{-1} a_i}{(v_i^{(t)})^{\alphap}} \notag\\
    &= \left(1+ \lambda\left(\left(1-\frac{1}{L}\right)^{\betap}-1\right)\right)^{-1} \left(1+ \lambda\left(1+\frac{\rho_i(v^{(t)})-1}{L}\right)^{\betap}-1\right)^{-\alphap} \cdot \rho_i(v^{(t)}). \label{eq:rho-u-second-factor}
\end{align}
We bound the two multiplicative factors in \eqref{eq:rho-u-second-factor} separately. For the first factor, we have
\begin{align} \label{eq:firstfactor}
\left(1+ \lambda\left(\left(1-\frac{1}{L}\right)^{\betap}-1\right)\right)^{-1} \leq  \left(1+ \left(\left(1-\frac{1}{L}\right)^{\betap}-1\right)\right)^{-1} \leq 2, 
\end{align}
where the first inequality uses that $\left(1-\frac1L\right)^{\betap}-1 \leq0$, and the last inequality uses that $\left(1-\frac1L\right)^{\betap} \geq \frac12$ since $L \geq 4\barbetap$. For the second factor in \eqref{eq:rho-u-second-factor}, we distinguish two cases: $\rho_i(v^{(t)}) \leq 1$ or $\rho_i(v^{(t)}) >1$. On the one hand, when $\rho_i(v^{(t)}) > 1$, this second factor is at most $1$. On the other hand, when $\rho_i(v^{(t)}) \leq 1$, we have $\left(1+\frac{\rho_i(v^{(t)})-1}{L}\right)^{\betap}-1 \leq 0$ and therefore this second factor increases when $\lambda$ increases. This shows that 
\begin{align*}
    \left(1+ \lambda\left(\left(1+\frac{\rho_i(v^{(t)})-1}{L}\right)^{\betap}-1\right)\right)^{-\alphap} 
    \leq \left(1+ \left(\left(1-\frac{1}{L}\right)^{\betap}-1\right)\right)^{-\alphap},
\end{align*}
which is at most $2^{\alphap} \leq 2$ by \eqref{eq:firstfactor} when $\alphap\leq 1$, and is at most
\begin{align*}
    \left(1-\frac{\betap}{L}\right)^{-\alphap}\leq \left(1-\frac{1}{L}\right)^{-1}\leq 2
\end{align*}
when $\alphap>1$. Together, this shows that $\rho_i(u) \leq 4 \rho_i(v^{(t)})$ for each $i \in [m]$, which gives $\rho_\max(u)\leq 4 \rho_\max(v^{(t)}) \leq 16\barbetap$. 
\end{proof}

\begin{proposition}\label{prop:primal_alg_func_decrease}
For any iteration $t\geq T/2$ in Algorithm~\ref{alg:relative-smoothness}, we have $\fobj_\primal(v^{(t)})-\fobj_\primal(v_{*})\leq \Delta$.
\end{proposition}

\begin{proof}
Given our choice of $L$, Lemma~\ref{lem:algorithm_rhomax_control} establishes that $\rho_\max(v^{(t)})\leq 4\barbetap$ for any iteration $t$ of Algorithm~\ref{alg:relative-smoothness}. Then, by Lemma~\ref{lem:fobj-relative-smoothness} and Lemma~\ref{lem:algorithm_rhomax_control_lambda}, it follows that $\fobj$ is $(1+16p\barbetap)\leq L$-smooth relative to $h_{\primal}$ between any two consecutive iterates $v^{(t)}$ and $v^{(t+1)}$, and is $1$-strongly convex relative to $h_{\primal}$. Hence, applying Proposition~\ref{prop:primal-gradient-scheme-convergence-local} gives
\begin{align*}
    \fobj_\primal(v^{(t)})-\fobj_\primal(v_{*})
    \leq \frac{(1-1/L)^t}{1-(1-1/L)^t} D_{h_{\primal}}(v_*,v^{(0)})\leq \Delta,
\end{align*}
where we used the fact that
\begin{align*}
D_{h_{\primal}}(v_*,v^{(0)})= h_{\primal}(v_*)-h_{\primal}(v^{(0)})-\langle\nabla h_{\primal}(v^{(0)}),v_*-v^{(0)}\rangle\leq 2m
\end{align*}
since $[v_*]_i\leq 1$ for any $i\in[m]$.
\end{proof}

\begin{remark}\label{rem:primal-coordinatewise-scaling-change}
By Lemma~\ref{lem:algorithm_rhomax_control} and the initialization, $\rho_{\max}(v^{(t)})\le 4\barbetap$ for every iteration $t$ in Algorithm~\ref{alg:relative-smoothness}. Therefore, for every coordinate $i\in[m]$,
\[
\left|\frac{v_i^{(t+1)}}{v_i^{(t)}}-1\right|
=
\left|
\left(1+\frac{\rho_i(v^{(t)})-1}{L}\right)^{1/\alphap}
-1
\right|
\le \frac12 .
\]
Consequently, the diagonal scaling $D^{(t)}=(V^{(t)})^{1/2}$ used in each leverage score computation satisfies
\[
\sqrt{\frac{1}{2}}\,D_{ii}^{(t)}
\le
D_{ii}^{(t+1)}
\le
\sqrt{\frac32}\,D_{ii}^{(t)},
\quad \forall i\in[m].
\]
\end{remark}

\subsection{One-sidedness property of the last iterate}\label{sec:primal_alg_one_sidedness_last_iterate}
Next, we show that the last iterate $v^{(\totaliters)}$ of Algorithm~\ref{alg:relative-smoothness} provides one-sided approximate Lewis weights. To facilitate this analysis, for any $i,j\in[m]$, we define $\gamma_{ij}$ to be the angle between the vectors $(\mA^\top V\mA)^{-1/2}v_i^{1/2}a_i$ and  $(\mA^\top V\mA)^{-1/2}v_j^{1/2}a_j$, i.e.,
\[
    \gamma_{ij}\coloneqq \arccos\bigg(\frac{a_i^\top v_i^{1/2}(A^\top VA)^{-1}v_j^{1/2}a_j}{\|(A^\top VA)^{-1/2}v_i^{1/2}a_i\|\cdot\|(A^\top VA)^{-1/2}v_j^{1/2}a_j\|}\bigg).
\]
The next three lemmas show that, for any iteration $t$ of Algorithm~\ref{alg:relative-smoothness} at which $\fobj_\primal(v^{(t)})$ is sufficiently small, the value of $|\rho_i(v^{(t)})-1|$ will decrease multiplicatively until it is below a certain threshold.

\begin{lemma}[Lemma 47 of \cite{lee2019solving}]\label{lem:sum-of-cosines}
For any vector $v\in\R^m_{>0}$ and $i\in[m]$, 
\[
\sum_{j\in[m]}\sigma_j(v)\cdot\cos^2(\gamma_{ij})=1,\qquad\forall i\in[m].
\]
\end{lemma}

\begin{lemma}\label{lem:contraction-strong-version}
    Let $v,v^+$ be vectors in $\R^m_{>0}$ such that $\frac{v}{2} \leq v^+\leq \frac {3v}2$ and define 
    \[
    \theta_i \coloneqq \sum_{j\in[m]} |v_j^+ -v_j| \cdot a_j^\top (\mA^\top V \mA)^{-1} a_j \cos(\gamma_{ij})^2 = \sum_{j\in[m]} \left|\frac{v_j^+}{v_j} -1\right| \cdot \sigma_j(v) \cos(\gamma_{ij})^2.
    \]
    Then for any $i$ we have 
    \[
    (1-3\theta_i) \left(\frac{v_i^+}{v_i}\right)^{-\alphap}  \leq \frac{\rho_i(v^+)}{\rho_i(v)} \leq  (1+3\theta_i)\left(\frac{v_i^+}{v_i}\right)^{-\alphap}.
    \]
\end{lemma}

\begin{proof}
Note that
\begin{align*}
\rho_i(v^{(+)}) &= \left[\frac{v^{(+)}_i}{v_i} \right]^{-\alphap} \frac{1}{[v_i]^{\alphap}}
a_i^\top (\mA^\top V^{(+)} \mA)^{-1} a_i = \left[\frac{v^{(+)}_i}{v_i} \right]^{-\alphap} \frac{a_i^\top (\mA^\top V^{(+)} \mA)^{-1} a_i}{a_i^\top (\mA^\top V \mA)^{-1} a_i}\cdot\rho_i(v).
\end{align*}
Hence, it suffices to bound $\frac{a_i^\top (\mA^\top V^{(+)} \mA)^{-1} a_i}{a_i^\top (\mA^\top V \mA)^{-1} a_i}$. Denote $\Delta\coloneqq \mA^\top(V^{(+)}-V)\mA$. Then, 
\begin{align*}
\frac{a_i^\top (\mA^\top V^{(+)} \mA)^{-1} a_i}{a_i^\top (\mA^\top V \mA)^{-1} a_i}
&=\frac{a_i^\top (\mA^\top V \mA+\Delta)^{-1} a_i}{a_i^\top (\mA^\top V \mA)^{-1} a_i}\\
&=\frac{a_i^\top (\mA^\top V\mA)^{-1/2}(I+\bar{\Delta})^{-1}(\mA^\top V\mA)^{-1/2} a_i}{a_i^\top (\mA^\top V \mA)^{-1} a_i},
\end{align*}
where
\begin{align*}
\bar{\Delta}=(\mA^\top V\mA)^{-1/2}\Delta(\mA^\top V\mA)^{-1/2} = (\mA^\top V\mA)^{-1/2}\mA^\top (V^+-V) \mA (\mA^\top V\mA)^{-1/2}.
\end{align*}
From the assumption $v/2 \leq v^+$ it follows that $\bar\Delta\succeq -I/2$. The latter in turn implies  
\begin{align*}
1-\bar\Delta \preceq (I+\bar\Delta)^{-1} \preceq 1-\bar\Delta + 2\bar \Delta^2.
\end{align*}
Therefore, we have the following chain of inequalities
\begin{align} 
\begin{aligned}
&\frac{a_i^\top (\mA^\top V\mA)^{-1/2}(1-\bar\Delta)(\mA^\top V\mA)^{-1/2} a_i}{a_i^\top (\mA^\top V \mA)^{-1} a_i} 
\leq \frac{a_i^\top (\mA^\top V^{(+)} \mA)^{-1} a_i}{a_i^\top (\mA^\top V \mA)^{-1} a_i} \\
&\qquad\qquad\qquad\qquad\qquad\qquad\qquad
\leq \frac{a_i^\top (\mA^\top V\mA)^{-1/2}(1-\bar\Delta + 2\bar\Delta^2)(\mA^\top V\mA)^{-1/2} a_i}{a_i^\top (\mA^\top V \mA)^{-1} a_i}.
\end{aligned}\label{eq:key chain}
\end{align}

We proceed by separately bounding the terms that depend linearly and quadratically on $\bar\Delta$. First, for the term that depends linearly on $\bar\Delta$, we have
\begin{align}
    |a_i^\top (\mA^\top V\mA)^{-1/2}\bar\Delta(\mA^\top V\mA)^{-1/2} a_i| &=| a_i^\top (\mA^\top V\mA)^{-1} \mA^\top (V^+-V) \mA (\mA^\top V\mA)^{-1} a_i |\nonumber\\
    &\leq \sum_{j\in[m]} |v_j^+-v_j| \left(a_i^\top (\mA^\top V\mA)^{-1} a_j\right)^2 \label{eq:linbound}\\
    &\leq a_i^\top(\mA^\top V\mA)^{-1} a_i \sum_{j\in[m]} |v_j^+-v_j| a_j^\top (\mA^\top V\mA)^{-1} a_j \cos(\gamma_{ij})^2 \nonumber
\end{align}
Second, for the term that depends quadratically on $\bar\Delta$, we have
\begin{align*}
    \bar \Delta^2 &= (\mA^\top V\mA)^{-1/2}\mA^\top (V^+-V) \mA (\mA^\top V\mA)^{-1}\mA^\top (V^+-V) \mA (\mA^\top V\mA)^{-1/2} \\
    &= (\mA^\top V\mA)^{-1/2}\mA^\top (V^+-V)  V^{-1/2} V^{1/2}\mA(\mA^\top V\mA)^{-1}\\
    &\qquad\qquad\cdot\mA^\top V^{1/2} V^{-1/2}(V^+-V) \mA (\mA^\top V\mA)^{-1/2} \\
    &\preceq (\mA^\top V\mA)^{-1/2}\mA^\top (V^+-V)  V^{-1/2} V^{-1/2}(V^+-V) \mA (\mA^\top V\mA)^{-1/2} \\
    &= (\mA^\top V\mA)^{-1/2}\mA^\top (V^+-V)^2  V^{-1} \mA (\mA^\top V\mA)^{-1/2}.
\end{align*}
This allows us to proceed as before and obtain
\begin{align} 
\begin{aligned}
    &|a_i^\top (\mA^\top V\mA)^{-1/2}\bar\Delta^2(\mA^\top V\mA)^{-1/2} a_i| \\
    &\qquad\leq | a_i^\top (\mA^\top V\mA)^{-1} \mA^\top (V^+-V)^2  V^{-1} \mA (\mA^\top V\mA)^{-1} a_i |\\
    &\qquad\leq \sum_{j\in[m]} \frac{|v_j^+-v_j|^2}{v_j} \left(a_i^\top (\mA^\top V\mA)^{-1} a_j\right)^2 \\
    &\qquad\leq a_i^\top(\mA^\top V\mA)^{-1} a_i \sum_{j\in[m]} \frac{|v_j^+-v_j|^2}{v_j} a_j^\top (\mA^\top V\mA)^{-1} a_j \cos(\gamma_{ij})^2 \\
    &\qquad\leq a_i^\top(\mA^\top V\mA)^{-1} a_i \sum_{j\in[m]} |v_j^+-v_j| a_j^\top (\mA^\top V\mA)^{-1} a_j \cos(\gamma_{ij})^2,
\end{aligned}\label{eq:quadbound}
\end{align}
where the last inequality uses the assumption $|v_j^+-v_j| \leq v_j$. Combining the above estimates \eqref{eq:linbound} and \eqref{eq:quadbound} with \eqref{eq:key chain} concludes the proof. 
\end{proof}

\begin{lemma}\label{lem:rho-contraction}
For any iteration $t$ in Algorithm~\ref{alg:relative-smoothness} and any $i\in[m]$, we have
\begin{align*}
    \rho_i(v^{(t+1)})\leq \max\left\{1+15L\theta(v^{(t)}),\left(1-\frac{\rho_i(v^{(t)})-1}{4L}\right)\rho_i(v^{(t)})\right\}
\end{align*}
where
\[
\theta(v)\coloneqq \frac{2}{\alphap L}\sqrt{\frac{30}\alphap(\fobj_\primal(v)-\fobj_\primal(v_{*}))},\quad\forall v\in\R^{m}_{\geq 0}.
\]
\end{lemma}
\begin{proof}
    By the choice of $L$ and Lemma~\ref{lem:algorithm_rhomax_control_lambda}, for any $i\in[m]$ we have
    \begin{align*}
    \left|\frac{v_i^{(t+1)}}{v_i^{(t)}}-1\right|
    =\left|\bigg(1+\frac{\rho_i(v^{(t)})-1}{L}\bigg)^{1/\alphap}-1\right|
    \leq \frac12.
    \end{align*}
    Then, invoking Lemma~\ref{lem:contraction-strong-version} gives
    \[
    (1-3\theta_i) \left(\frac{v_i^{(t+1)}}{v_i^{(t)}}\right)^{-\alphap}  \leq \frac{\rho_i(v^{(t+1)})}{\rho_i(v^{(t)})} \leq  (1+3\theta_i)\left(\frac{v_i^{(t+1)}}{v_i^{(t)}}\right)^{-\alphap},
    \]
    where
    \begin{align*}
        \theta_i 
        &= \sum_{j\in[m]} \bigg|\frac{v_j^{(t+1)}}{v_j^{(t)}} -1\bigg| \cdot \sigma_j(v^{(t)}) \cos(\gamma_{ij})^2\\
        &=\sum_{j\in[m]}\bigg|\bigg(1+\frac{\rho_j(v^{(t)})-1}{L}\bigg)^{1/\alphap}-1\bigg|\cdot\sigma_j(v^{(t)}) \cos(\gamma_{ij})^2\\
        &\leq \frac{2}{\alphap L}\sum_{j\in[m]}|\rho_j(v^{(t)})-1|\cdot\sigma_j(v^{(t)}) \cos(\gamma_{ij})^2,
    \end{align*}
    where the last inequality uses the fact that
    \[
    \frac{\rho_j(v^{(t)})-1}{L}
    \leq \frac{\rho_\max(v^{(t)})-1}{L}\leq \frac\alphap 4.
    \]
    By Cauchy-Schwarz inequality, we have
    \begin{align*}
        &\sum_{j\in[m]}|\rho_j(v^{(t)})-1|\cdot\sigma_j(v^{(t)}) \cos(\gamma_{ij})^2\\
        &\qquad\qquad\leq\sqrt{\sum_{j\in[m]}\sigma_j(v^{(t)})(\rho_j(v^{(t)})-1)^2}\cdot\sqrt{\sum_{j\in[m]}\sigma_j(v^{(t)})\cos^2(\gamma_{ij})}\\
        &\qquad\qquad\leq\sqrt{\sum_{j\in[m]}\sigma_j(v^{(t)})(\rho_j(v^{(t)})-1)^2}\\
        &\qquad\qquad\leq \sqrt{\rho_\max(v^{(t)})}\sqrt{\sum_{j\in[m]}[v_j^{(t)}]^{1+\alphap}(\rho_j(v^{(t)})-1)^2}\\
        &\qquad\qquad\leq \sqrt{\frac {30}\alphap(\fobj_\primal(v^{(t)})-\fobj_\primal(v_{*}))},
    \end{align*}
    where the second inequality uses Lemma~\ref{lem:sum-of-cosines} and the last inequality uses Lemma~\ref{lem:fobj_primal_gap_lower}. Hence, we have $\theta_i\leq \theta(v^{(t)})$ for all $i\in[m]$, which leads to
    \begin{align}
    \rho_i(v^{(t+1)})
    &\leq \frac{1+3\theta(v^{(t)})}{1+\frac{\rho_i(v^{(t)})-1}{L}}\cdot\rho_i(v^{(t)})\nonumber\\
    &\leq (1+3\theta(v^{(t)}))\left(1-\frac{\rho_i(v^{(t)})-1}{2L}\right)\rho_i(v^{(t)}).\label{eq:rounding-inequality}
    \end{align}
    Note that the function $\phi(x)\coloneqq \left(1-\frac{x-1}{2L}\right)x$ is monotonically increasing in $[0,1]$. Thus for any $i$ with $\rho_i(v^{(t)})\leq 1$, the value of \eqref{eq:rounding-inequality} is at most $1+3\theta(v^{(t)})$. Otherwise, we have
    \begin{align*}
        \rho_i(v^{(t+1)})\leq \max\left\{1+15L\theta(v^{(t)}),\left(1-\frac{\rho_i(v^{(t)})-1}{4L}\right)\rho_i(v^{(t)})\right\}.
    \end{align*}
\end{proof}

\begin{lemma}\label{lem:contraction-sequence-convergence}
For any $0<\betap\leq 1$ and $L\geq 1$, let $\{\zeta^{(t)}\}_{t\ge 0}$ be a sequence satisfying 
\begin{align}
    \zeta^{(t+1)}
    \leq\max\left\{1+\betap,\bigg(1-\frac{\zeta^{(t)}-1}{4L}\bigg)\zeta^{(t)}\right\},
    \quad \forall t\in\mathbb N\label{eq:recursive-upper}
\end{align}
with $0<\zeta^{(0)}\le L$ and $L\ge 1$. Then, there exists a finite index $\overline t=\big\lceil 4L
\ln(\max\{\zeta^{(0)},1+\betap\}/\betap)\big\rceil$ such that
\[
\zeta^{(t)} \leq 1+\betap,\quad \forall t\ge \overline t.
\]
\end{lemma}
\begin{proof}
First note that for every $t\ge 0$, if $\zeta^{(t)}\le 1+\betap$, we have
\[
\zeta^{(t+1)} \le \max\left\{1+\betap,\left(1-\frac{\zeta^{(t)}-1}{4L}\right)\zeta^{(t)}\right\}\le 1+\betap.
\]
Hence, it suffices to prove that $\zeta^{(\bar t)}\leq 1+\betap$.

Assume the contrary, i.e, there exists a sequence $\{\zeta^{(t)}\}_{t\geq 0}$ satisfying \eqref{eq:recursive-upper} with $\zeta^{(\bar t)}> 1+\betap$. Then, for any $t<\bar t$ we have $\zeta^{(t)}> 1+\betap$, and thus
\[
\zeta^{(t+1)}-1
\le \left(1-\frac{\zeta^{(t)}-1}{4L}\right)\zeta^{(t)}-1
\le \left(1-\frac{1}{4L}\right)(\zeta^{(t)}-1),
\]
which leads to
\[
\zeta^{(\bar t)} \le 1+\left(1-\frac1{4L}\right)^{\bar t} (\zeta^{(0)}-1)
\le 1+ \exp\!\left(-\frac{\bar t}{4L}\right)\zeta^{(0)}\leq 1+\betap,
\]
contradiction. Therefore, we can conclude that we have $\zeta^{(t)}\leq 1+\betap$ for any sequence $\{\zeta^{(t)}\}_{t\geq 0}$ satisfying \eqref{eq:recursive-upper} and any $t\ge \overline t$.
\end{proof}

\begin{lemma}\label{lem:function_value_gap_to_one_sidedness}
    For any $0<\epsilon\leq 1/2$ and any $v\in\R^m_{>0}$ satisfying $\rho_\max(v)\leq 1+\epsilon$ and $\fobj_\primal(v)-\fobj_\primal(v_{*})\leq \epsilon^3n/(384\baralphap)$, the vector $w=v^{1+\alphap}$ is a one-sided $\epsilon$-approximation of $\sigma_p(A)$. 
\end{lemma}

\begin{proof}
    We directly have $\sigma(w^{\frac12-\frac1p})\leq (1+\epsilon)w$ since $\rho_\max(w^{\frac12-\frac1p})=\rho_\max(v)\leq 1+\epsilon$. As for the upper bound on $\|w\|_1$, we denote
    \begin{align*}
        \Gamma\coloneqq\{i\in[m]\mid \rho_i(v)\geq 1-\epsilon/4\}.
    \end{align*}
    Then, $\|w\|_1=\sum_{i\in\Gamma}w_i+\sum_{i\in[m]\backslash\Gamma}w_i$, where we have
    \begin{align*}
        \sum_{i\in\Gamma}w_i
        \leq \frac{1}{1-\epsilon/4}\sum_{i\in\Gamma}\sigma_i(w^{\frac12-\frac1p})
        \leq \left(1+\frac\epsilon2\right)\sum_{i\in[m]}\sigma_i(w^{\frac12-\frac1p})
        \leq \left(1+\frac\epsilon2\right)n
    \end{align*}
    and
    \begin{align*}
        \sum_{i\in[m]\backslash\Gamma}w_i
        =\sum_{i\in[m]\backslash\Gamma}v_i^{1+\alphap}
        &\leq\frac{32}{\epsilon^2}\sum_{i\in[m]\backslash\Gamma}v_i^{1+\alphap}\frac{(\rho_i(v)-1)^2}{\rho_i(w)+1}\\
        &\leq \frac{192\baralphap}{\epsilon^2}(\fobj_\primal(v)-\fobj_\primal(v_{*}))\leq\frac{\epsilon n}{2},
    \end{align*}
    where the second inequality uses Lemma~\ref{lem:fobj_primal_gap_lower}. Hence, we can conclude that $\|w\|_1\leq (1+\epsilon)n$ and thus $w$ is a one-sided $\epsilon$-approximate Lewis weight vector.
\end{proof}

\theoapproxone*

\begin{proof}
By Proposition~\ref{prop:primal_alg_func_decrease}, for any iteration $t\geq \totaliters/2$ we have $\fobj_\primal(v^{(t)})-\fobj_\primal(v_{*})\leq \Delta$. Then, invoking Lemma~\ref{lem:rho-contraction} gives
\[
\rho_i(v^{(t+1)})\leq \max\left\{1+\hat\epsilon,\left(1-\frac{\rho_i(v^{(t)})-1}{4L}\right)\rho_i(v^{(t)})\right\}
\]
for any $i\in[m]$ and $t\geq \totaliters/2$. Therefore, by Lemma~\ref{lem:contraction-sequence-convergence}, we have $\rho_i(v^{(T)})\leq 1+\hat\epsilon$ for any $i\in[m]$, or equivalently, $\rho_\max(v^{(T)})\leq 1+\hat\epsilon$. Using Lemma~\ref{lem:function_value_gap_to_one_sidedness}, we can conclude that $w$ is a one-sided $\hat\epsilon$-approximation of $\sigma_p(A)$. 
\end{proof}

\begin{proof}[Proof of \Cref{thm:primal-algorithm-guarantee}]
By \Cref{thm:primal_algorithm_one_sided_guarantee}, the vector $w=[v^{(T)}]^{1+\alphap}$ is a one-sided $\hat\epsilon$-approximation of $\sigma_p(A)$. Applying \Cref{thm:one-sided-to-multiplicative} to $w$ produces precisely the vector $\hw$ returned by \Cref{alg:relative-smoothness}, and shows that $\hw$ is an $\epsilon$-estimate of $\sigma_p(A)$. The claimed iteration bound follows from \Cref{thm:primal_algorithm_one_sided_guarantee} and our choice of $\hat\epsilon$.
\end{proof}

\section{Conversion between different approximation guarantees}\label{sec:one_sided_to_others}

In this section, we develop several reductions between the different notions of Lewis weight approximation used throughout the paper, and apply them to obtain an improved analysis of an existing algorithm by Lee~\cite{lee16}. In particular, we first show how to convert one-sided approximations of $\sigma_p(A)$ to two-sided approximations in Section~\ref{sec:one-sided-to-two-sided}, and then how to convert them to estimates of $\sigma_p(A)$ in Section~\ref{sec:one-sided-to-multiplicative}. In Section~\ref{sec:approx-min-to-lewis}, we show how to convert approximate minimizers of $\fobj_\primal$ to two-sided approximations. Finally, in Section~\ref{sec:improved-analysis-Lee's-alg}, we present an improved analysis of a variant of the algorithm in \cite{lee16}; the full guarantee is stated in Theorem~\ref{thm:algo}. Throughout this section, we let $H\coloneqq A^\top W^{\frac12-\frac1p}A$, $\widehat H\coloneqq A^\top\widehat W^{\frac12-\frac1p}A$, and $\hat\rho\coloneqq \rho(\widehat w^{\frac12-\frac1p})\in\R^m_{>0}$.

\subsection{From one-sided approximations to two-sided approximations}\label{sec:one-sided-to-two-sided}

Here we establish the following conversion from one-sided to two-sided approximations, showing that if $w$ is a one-sided $\eps_{\mathrm{one}}$-approximation, then $\widehat{w} \coloneqq \sigma(w^{\frac12-\frac1p})^{\frac{p}{2}}/w^{\betap}$ is a two-sided $\eps_{\mathrm{two}}$-approximation for a suitable $\eps_{\mathrm{two}}$.

\theoonetotwo*

We first prove a key lemma that allows us to compare quadratic forms associated with $\widehat{w}$ and~$w$.

\begin{lemma}
\label{lem:mult_diff}
For any $w\in\R_{>0}^{m}$ and 
$\epsilon_{\mathrm{two}} \coloneqq \big\|\widehat{w} - \sigma(w^{\frac12-\frac1p}) \big\|_1$, we have
\begin{align*}
\|H^{-1/2}(\widehat H-H)H^{-1/2}\|_1
\leq\big\|\widehat{w} - \sigma(w^{\frac12-\frac1p}) \big\|_1
=\epsilon_{\mathrm{two}},
\end{align*}
and consequently,
\[
(1-\epsilon_{\mathrm{two}})H
\preceq \widehat H
\preceq (1+\epsilon_{\mathrm{two}})H.
\]
\end{lemma}
\begin{proof}
Let $\Delta\defeq \widehat{w}^{1-\frac2p}-w^{1-\frac2p}$ and let $\Delta_{+}\defeq\max\{\Delta,\vzero\}$
and $\Delta_{-}\defeq\max\{-\Delta,\vzero\}$ entrywise so that $\Delta_{+},\Delta_{-}\in\R_{\geq0}^{m}$,
$\Delta=\Delta_{+}-\Delta_{-}$, and $\norm{\Delta}_{1}=\norm{\Delta_{+}}_{1}+\norm{\Delta_{-}}_{1}$. Then, we have
\begin{align*}
\normFull{\mh^{-1/2}(\widehat H-H)\mh^{-1/2}}_{1} 
&\leq\normFull{\mh^{-1/2}\big(\ma^{\top}\Delta_{+}\ma\big)\mh^{-1/2}}_{1}
+\normFull{\mh^{-1/2}\big(\ma^{\top}\Delta_{-}\ma\big)\mh^{-1/2}}_{1}\,.
\end{align*}
We then use the fact $\Delta_+$ is positive semidefinite, and therefore so is $\mh^{-1/2}\left(\ma^{\top}\Delta_{+}\ma\right)\mh^{-1/2}$, to upper bound the spectral norm by the trace: 
\begin{align*}
\normFull{\mh^{-1/2}\big(\ma^{\top}\Delta_{+}\ma\big)\mh^{-1/2}}_{1}
&=\tr\left[\mh^{-1/2}\big(\ma^{\top}\Delta_{+}\ma\big)\mh^{-1/2}\right]\\
&=\sum_{i\in[m]}\left[\Delta_{+}\right]_{i}\left[\ma(\ma^{\top}\mw^{\frac12-\frac1p}\ma)^{-1}\ma^{\top}\right]_{ii}\
=\sum_{i\in[m]}\left[\Delta_{+}\right]_{i}\!\cdot\!\frac{\sigma_i(w^{\frac12-\frac1p})}{w_i^{1-\frac2p}}\,.
\end{align*}
By symmetry, the same bound holds for $\Delta_{-}$.
Combining the two bounds yields that 
\[
\normFull{\mh^{-1/2}(\ma^{\top}\Delta\ma)\mh^{-1/2}}_1
\leq\sum_{i\in[m]}\left|\Delta_{i}\right|\cdot\frac{\sigma_i(w^{\frac12-\frac1p})}{w_{i}^{1-\frac2p}}
\]
Recalling that $\Delta = \widehat{w}^{1-\frac2p}- w^{1-\frac2p}$ we obtain the desired upper bound:
\[
\sum_{i\in[m]}\left|\Delta_{i}\right|\cdot\frac{\sigma_i(w^{\frac12-\frac1p})}{w_{i}^{1-\frac2p}}
=
\big\| \widehat{w} - \sigma(w^{\frac12-\frac1p})\big\|_1=\epsilon_{\mathrm{two}}\,.
\]
Finally, we conclude that
\begin{align*}
    \normFull{\mh^{-1/2}(\ma^{\top}\Delta\ma)\mh^{-1/2}}
    \leq \normFull{\mh^{-1/2}(\ma^{\top}\Delta\ma)\mh^{-1/2}}_1 
    \leq \epsilon_{\mathrm{two}}
\end{align*}
and thus $(1-\epsilon_{\mathrm{two}})H
\preceq \widehat H
\preceq (1+\epsilon_{\mathrm{two}})H$.
\end{proof}
Next, we show how to bound the approximation factor $\epsilon_{\mathrm{two}} = \big\| \widehat{w} - \sigma(w^{\frac12-\frac1p})\big\|_1$ when $w$ is a one-sided approximation of $\sigma_p(A)$.
\begin{lemma}\label{lem:distances} 
If $w \in \R^m_{> 0}$ is a one-sided $\epsilon_{\mathrm{one}}$-approximate $\ell_p$-Lewis weight of $\ma$ for  $p > 2$, then 
\[
\big\| \widehat{w} - \sigma(w^{\frac12-\frac1p})\big\|_1
\leq 3\barbetap n\epsilon_{\mathrm{one}} (1 + \epsilon_{\mathrm{one}})^{\barbetap}.
\]
\end{lemma}

\begin{proof}
If $p\geq 4$ and $\betap\geq 1$, since $\widehat{w}$ and $\sigma(w^{\frac12-\frac1p})$ are non-negative, it follows that
\begin{align*}
\big\| \widehat{w} - \sigma(w^{\frac12-\frac1p})\big\|_1
&= \sum_{i \in [m]} \sigma_i(w^{\frac12-\frac1p}) \left| \bigg(\frac{\sigma_i(w^{\frac12-\frac1p})}{ w_i} \bigg)^{\betap} - 1 \right|\\
&\leq \betap(1+ \epsilon_{\mathrm{one}})^{\betap}
\norm{\sigma(w^{\frac12-\frac1p})-w}_1
\, ,
\end{align*}
where in the last step we used the fact that
\[
|1 - x^c| = \left| \int_1^{x} c \cdot  y^{c - 1} \d y \right|
\leq 
\left| \int_1^{x} c\cdot  \max\{1,x\}^{c - 1} \d y \right|
= c \cdot \max\{1,x\}^{c - 1} \cdot | 1 - x|
\]
for any $x\geq 0$ and $c>1$, and that $\sigma(w^{\frac12-\frac1p}) \leq (1+\epsilon_{\mathrm{one}}) w$. Otherwise, if $2\leq p<4$ and $0<\betap<1$, we have
\begin{align*}
\big\| \widehat{w} - \sigma(w^{\frac12-\frac1p})\big\|_1
&= \sum_{i \in [m]} \sigma_i(w^{\frac12-\frac1p}) \left| \bigg(\frac{\sigma_i(w^{\frac12-\frac1p})}{ w_i} \bigg)^{\betap} - 1 \right|\\
&\leq \sum_{i \in [m]} \frac{\sigma_i(w^{\frac12-\frac1p})}{w_i} \left| \sigma_i(w^{\frac12-\frac1p})-w_i \right|\leq (1+\epsilon_{\mathrm{one}})\norm{\sigma(w^{\frac12-\frac1p})-w}_1.
\end{align*}
The result then follows from the fact that $\norm{w- \sigma(w^{\frac12-\frac1p})}_1 \leq 3\epsilon_{\mathrm{one}} \norm{\sigma(w)}_1 = 3 \epsilon_{\mathrm{one}} n$, by Lemma~\ref{lem:abs_diff_2}.
\end{proof}

Now we are ready to prove Theorem~\ref{thm:one-sided-to-two-sided}.

\begin{proof}[Proof of Theorem~\ref{thm:one-sided-to-two-sided}]
For any $i\in[m]$, we have
\begin{align*}
\frac{\sigma(\widehat w^{\frac12-\frac1p})_{i}}{\widehat w_{i}} 
& =\widehat w_{i}^{-\frac2p}\cdot\left[\ma(\ma^{\top}\widehat W^{1-\frac2p}\ma)^{-1}\ma^{\top}\right]_{ii}
=\frac{\left[\ma(\ma^{\top}\widehat W^{1-\frac2p}\ma)^{-1}\ma^{\top}\right]_{ii}}{\left[\ma(\ma^{\top}\mw^{1-\frac2p}\ma)^{-1}\ma^{\top}\right]_{ii}}\,.
\end{align*}
We wish to lower and upper bound this fraction by $1/(1+\epsilon_{\mathrm{two}})$ and $1/(1-\epsilon_{\mathrm{two}})$, respectively.
For this it suffices to prove that
\[
\frac{1}{1+\epsilon_{\mathrm{two}}} \big( \ma^{\top}\mw^{1-\frac2p}\ma \big)^{-1}
\preceq \big( \ma^{\top}\widehat W^{1-\frac2p}\ma \big)^{-1}
\preceq \frac{1}{1-\epsilon_{\mathrm{two}}} \big( \ma^{\top}\mw^{1-\frac2p}\ma \big)^{-1},
\]
which is equivalent to showing that
\[
(1-\epsilon_{\mathrm{two}}) \ma^{\top}\mw^{1-\frac2p}\ma 
\preceq \ma^{\top}\widehat W^{1-\frac2p}\ma
\preceq (1+\epsilon_{\mathrm{two}}) \ma^{\top}\mw^{1-\frac2p}\ma.
\]
By Lemma~\ref{lem:mult_diff} and Lemma~\ref{lem:distances} this holds for $\epsilon_{\mathrm{two}} = \| \widehat{w} - \sigma(w^{\frac12-\frac1p})\|_1 \leq 3\barbetap n\epsilon_{\mathrm{one}} (1 + \epsilon_{\mathrm{one}})^{\barbetap}$.
\end{proof}

\subsection{From one-sided approximations to estimates of Lewis weights}\label{sec:one-sided-to-multiplicative}

We next show that, under a sufficiently strong one-sided approximation guarantee, the same postprocessing as in \Cref{sec:one-sided-to-two-sided} gives an estimate of $\sigma_p(A)$.

\theoonetomulti*

\begin{lemma}\label{lem:ln-sigma-norm} 
Let $\widehat\Sigma=\diag(\sigma_i(\widehat w^{\frac12-\frac1p}))$ and $\epsilon_{\mathrm{two}} = \| \widehat{w} - \sigma(w^{\frac12-\frac1p})\|_1$. If $ \epsilon_{\mathrm{two}}\leq 1/8$, then $\widehat w$ satisfies $\|\ln \rho(\widehat w^{\frac12-\frac1p})\|_{\widehat\Sigma}\leq 4\epsilon_{\mathrm{two}}^{1/2}$.
\end{lemma}
\begin{proof}
\begin{align*}
\rho_i(\widehat{w}^{\frac12-\frac1p})
&=\widehat{w}_i^{-\frac2p}a_i^\top (A^\top \widehat{W}^{1-\frac2p}A )^{-1}a_i
=\frac{a_i^\top( A ^\top \widehat{W}^{1-\frac2p} A )^{-1}a_i}{a_i^\top( A ^\top W^{1-\frac2p} A )^{-1}a_i}
=\frac{a_i^\top  \widehat{H}^{-1}a_i}{a_i^\top  H^{-1}a_i},
\end{align*}
which leads to
\begin{align*}
\rho_i(\widehat{w}^{\frac12-\frac1p})^{-1}=\frac{a_i^\top  \widehat{H}^{-\frac{1}{2}} \widehat{H}^{\frac{1}{2}} H^{-1} \widehat{H}^{\frac{1}{2}} \widehat{H}^{-\frac{1}{2}}a_i}{a_i^\top  \widehat{H}^{-1}a_i}
\end{align*}
where 
\begin{align*}
 \widehat{H}^{\frac{1}{2}} H^{-1} \widehat{H}^{\frac{1}{2}}
=\left( \widehat{H}^{-\frac{1}{2}} H \widehat{H}^{-\frac{1}{2}}\right)^{-1}=\left( I + \widehat{H}^{-\frac{1}{2}}( H- \widehat{H}) \widehat{H}^{-\frac{1}{2}}\right)^{-1}.
\end{align*}
By Lemma~\ref{lem:mult_diff} and Lemma~\ref{lem:matrix-ell1-distance-2}, we have
\begin{align*}
\left\| \widehat{H}^{-\frac{1}{2}}( H- \widehat{H}) \widehat{H}^{-\frac{1}{2}}\right\|_1\leq 2 \epsilon_{\mathrm{two}}.
\end{align*}
Denote
\begin{align*}
\Delta\coloneqq\left( I + \widehat{H}^{-\frac{1}{2}}( H- \widehat{H}) \widehat{H}^{-\frac{1}{2}}\right)^{-1}- I .
\end{align*}
Then, $\|\Delta\|_1\leq 2 \epsilon_{\mathrm{two}}$ by Lemma~\ref{lem:inverse-ell1-distance}. Define $\Delta_+\coloneqq \max\{\Delta,0\}$ and $\Delta_-\coloneqq \max\{-\Delta,0\}$ entrywise, and define a new positive semidefinite matrix $\bar{\Delta}\coloneqq\Delta_+-\Delta_-$ that satisfies
\begin{align*}
\|\bar{\Delta}\|_1\leq\|\Delta_+\|_1+\|\Delta_-\|_1\leq 2\|\Delta\|_1\leq 4 \epsilon_{\mathrm{two}}.
\end{align*}
Moreover, we have
\begin{align*}
\left|\rho_i(\widehat{w}^{\frac12-\frac1p})^{-1}-1\right|\leq\frac{\| \widehat{H}^{-\frac{1}{2}}a_i\|_{\bar{\Delta}}^2}{\| \widehat{H}^{-\frac{1}{2}}a_i\|_2^2}\leq \|\bar{\Delta}\|_2\leq4 \epsilon_{\mathrm{two}}\leq\frac{1}{2},
\end{align*}
which leads to
\begin{align*}
|\ln(\rho_i(\widehat w^{\frac12-\frac1p}))|
\leq \frac{2\| \widehat{H}^{-\frac{1}{2}}a_i\|_{\bar{\Delta}}^2}{\| \widehat{H}^{-\frac{1}{2}}a_i\|_2^2}
=\frac{2\| \widehat{H}^{-\frac{1}{2}}\widehat{w}_i^{\frac{1}{2}-\frac1p}a_i\|_{\bar{\Delta}}^2}{\| \widehat{H}^{-\frac{1}{2}}\widehat{w}_i^{\frac{1}{2}-\frac1p}a_i\|_2^2},
\end{align*}
and 
\begin{align*}
\|\ln \rho(\widehat w^{\frac12-\frac1p})\|_{\widehat\Sigma}^2
=\sum_{i\in[m]}\sigma_i(\widehat{w}^{\frac12-\frac1p})|\ln(\rho_i(\widehat w^{\frac12-\frac1p}))|^2
\leq 4\sum_{i\in[m]}\sigma_i(\widehat{w}^{\frac12-\frac1p})\cdot\frac{\| \widehat{H}^{-\frac{1}{2}}\widehat{w}_i^{\frac{1}{2}-\frac1p}a_i\|_{\bar{\Delta}}^4}{\| \widehat{H}^{-\frac{1}{2}}\widehat{w}_i^{\frac{1}{2}-\frac1p}a_i\|_2^4}.
\end{align*}
Denote $\hat{ A }=\widehat{W}^{\frac12-\frac1p} A  \widehat{H}^{-\frac{1}{2}}$. Then by Lemma~\ref{lem:leverage-score-one-norm}, 
\begin{align*}
\|\ln \rho(\widehat w^{\frac12-\frac1p})\|_{\Sigma}^2
&\leq 4\sum_{i\in[m]}\sigma_i(\widehat{w}^{\frac12-\frac1p} )\cdot\left(\frac{\hat{a}_i^\top\bar{\Delta}\hat{a}_i}{\hat{a}_i^\top\hat{a}_i}\right)^2
\leq
4\sum_{i\in[m]}\sigma_i(\widehat{w}^{\frac12-\frac1p})\cdot\left(\frac{\hat{a}_i^\top\bar{\Delta}\hat{a}_i}{\hat{a}_i^\top\hat{a}_i}\right)\\
&= 4\tr[\hat{ A }\bar{\Delta}\hat{ A }^\top]\leq\|\bar{\Delta}\|_1\leq 16 \epsilon_{\mathrm{two}}
\end{align*}
here we use that $\left(\frac{\hat{a}_i^\top\bar{\Delta}\hat{a}_i}{\hat{a}_i^\top\hat{a}_i}\right) \leq 1$ and that $\sigma_i(\widehat w^{\frac12-\frac1p}) = \hat{a}_i^\top \hat a_i$.
\end{proof}

\begin{lemma}[Lemma 14 and Claim 1 of \cite{fazel22}]\label{lem:integral}
Consider $\hat{v}\coloneqq\widehat w^{1-\frac2p}$ which satisfies $|\rho_i(\hat{v})-1|\leq1/2$ for any $i$, define
\begin{align}
\label{eq:hatv_t-def}
\hat{v}(t)=\underset{v\in\R^m_{>0}}{\mathrm{argmin}}f_t(v)\coloneqq-\log\det\big(A^\top V A\big)+\frac1{1+\alphap}\sum_{i=1}^m \rho_i^t(\hat{v})v_i^{1+\alphap},\quad\forall t\in[0,1].
\end{align}
Then, we have $\hat{v}(1)=\hat{v}$, $\hat{v}(0)=v_{*}$, and
\begin{align*}
    \left\|\frac{\mathrm{d}}{\mathrm{d} t}\ln\left(\frac{\hat{v}(t)}{\hat{v}(1)}\right)\right\|_\infty
    \leq\frac{\|\ln \rho(\hat{v})\|_\infty}{\alphap}
    +\frac{\|\ln \rho(\hat{v})\|_{\Sigma(\hat{v}(t))}}{\alphap^2},
\end{align*}
\end{lemma} 

\begin{lemma}\label{lem:multiplicative-leverage-score}
For any $\xi\in\R^m_{>0}$ that satisfies $\|\ln \xi\|_{\infty}=\gamma\leq 1/4$, we have
\begin{align*}
\left\|\frac{\sigma( \diag(\xi)\ma)}{\sigma(\ma)}\right\|_{\infty}\leq 1+8\gamma.
\end{align*}
\end{lemma}
\begin{proof}
Given that $\|\ln \xi\|_{\infty}=\gamma\leq 1/4$, we have $\frac{1}{1+2\gamma}\leq \xi_i\leq 1+2\gamma$ for all $i\in[m]$. Then for each $i$, we have
\begin{align*}
\sigma_i( \xi \ma)
&=\xi_i^2a_i^\top(\ma^\top \diag(\xi)^2\ma)^{-1}a_i\\
&\leq (1+2\gamma)^2a_i^\top\left(\frac{\ma^\top\ma}{(1+2\gamma)^2}\right)^{-1}a_i\leq (1+2\gamma)^4\sigma_i(\ma)\leq (1+8\gamma)\sigma_i(\ma).
\end{align*}
\end{proof}

\begin{proof}[Proof of \cref{thm:one-sided-to-multiplicative}]
By Theorem~\ref{thm:one-sided-to-two-sided}, we have that $\widehat w$ is an $\epsilon_{\mathrm{two}}$-two-sided approximation of $\sigma_p(A)$ for some $
\epsilon_{\mathrm{two}}$ satisfying
\begin{align}
\label{eq:multiplicative-proof-two-sided-property}
\epsilon_{\mathrm{two}}
\leq3\barbetap n\epsilon_{\mathrm{one}}(1 + \epsilon_{\mathrm{one}})^{\barbetap}
\leq 6\barbetap n\epsilon_{\mathrm{one}}
\leq \frac18
\end{align}
by our choice of $\epsilon_{\mathrm{one}}$. Assume $\hat v=\widehat{w}^{1-\frac2p}$ satisfies $\|\ln(\hat{v}/v_{*})\|_{\infty}\leq\frac{1}{4}$, which will be justified later. Consider the vector function $\hat v(t)$ defined in \eqref{eq:hatv_t-def}, by \cref{thm:one-sided-to-two-sided} and Lemma~\ref{lem:integral}, we have
\begin{align*}
    \left\|\frac{\mathrm{d}}{\mathrm{d} t}\ln\left(\frac{\hat v(t)}{\hat v(1)}\right)\right\|_\infty
    \leq\frac{\|\ln \hat\rho\|_\infty}{\alphap}
    +\frac{\|\ln \hat\rho\|_{\Sigma(\hat v(t))}}{\alphap^2}
    \leq\frac{2\epsilon_{\mathrm{two}}}{\alphap}+\frac{\|\ln \hat\rho\|_{\Sigma(\hat v(t))}}{\alphap^2}.
\end{align*}
By Lemma~\ref{lem:multiplicative-leverage-score},
\begin{align*}
\|\ln \hat\rho\|_{\Sigma(\hat v(t))}
\leq \|\ln \hat\rho\|_{\Sigma(\hat{v})}\cdot\left\|\frac{\sigma(\hat v(t))}{\sigma(\hat{v}(1))}\right\|_{\infty}
\leq  \|\ln \hat\rho\|_{\Sigma(\hat{v})}\left(1+8\left\|\ln\left(\frac{\hat{v}(t)}{\hat{v}(1)}\right)\right\|_{\infty}\right).
\end{align*}
Moreover, since
\[
\left\|\ln\left(\frac{\hat{v}(t)}{\hat{v}(1)}\right)\right\|_{\infty}\leq \left\|\ln\left(\frac{\hat{v}(0)}{\hat{v}(1)}\right)\right\|_{\infty}\leq\frac14,
\]
and $\|\ln \hat\rho\|_{\Sigma(\hat{v})}\leq 4 \epsilon_{\mathrm{two}}$ by Lemma~\ref{lem:ln-sigma-norm}, we obtain
\begin{align*}
    \left\|\frac{\mathrm{d}}{\mathrm{d} t}\ln\left(\frac{\hat{v}(t)}{\hat{v}(1)}\right)\right\|_\infty
    &\leq\frac{p-2}{2} \epsilon_{\mathrm{two}}+4 \epsilon_{\mathrm{two}}^{1/2}\left(\frac{p-2}{2}\right)^2\left(1+8\left\|\ln\left(\frac{\hat{v}(t)}{\hat{v}(1)}\right)\right\|_{\infty}\right)\\
    &\leq(p-2) \epsilon_{\mathrm{two}}+4(p-2)^2 \epsilon_{\mathrm{two}}^{1/2}\leq\frac{1}{4}.
\end{align*}
Integrating over $t\in[0,1]$ gives
\begin{align*}
\left\|\ln\left(\frac{\hat{v}(t)}{\hat{v}(1)}\right)\right\|_{\infty}
\leq (p-2)\left(4p-7\right) \epsilon_{\mathrm{two}}^{1/2}
\leq (p-2)\left(4p-7\right)\sqrt{6\barbetap n\epsilon_{\mathrm{one}}}
\leq\frac{1}{4},\quad\forall t\in[0,1],
\end{align*}
where the second inequality uses \eqref{eq:multiplicative-proof-two-sided-property}. This verifies the assumption made at the beginning of the proof. Since $\hat v(0)=v_{*}$, it follows that
\begin{align*}
\exp\big(-(p-2)\left(4p-7\right) \epsilon_{\mathrm{two}}^{1/2}\big)v_{*}
\leq \hat v
\leq\exp\big((p-2)\left(4p-7\right) \epsilon_{\mathrm{two}}^{1/2}\big)v_{*}
\end{align*}
and thus $(1-\epsilon_{\mathrm{est}})\sigma_p(A)\leq \widehat w\leq (1+\epsilon_{\mathrm{est}})\sigma_p(A)$.
\end{proof}

\subsection{From approximate minimizers of \texorpdfstring{$\fobj_\primal$}{Fvec} to two-sided approximations}\label{sec:approx-min-to-lewis}

Here we give a postprocessing step that turns an approximate minimizer of $\mathcal{F}_{\mathrm{vec}}$ satisfying a mild upper bound on $\rho_{\max}$ into a two-sided approximation of $\sigma_p(A)$.

\theoapproxtotwo*

For comparison, prior work established the following conversion from approximate optimality to estimates. Ours achieves two-sided approximations instead of estimates, but without polynomial overhead in terms of the dimension.
\begin{lemma}[Lemma 1 of \cite{fazel22}]\label{lem:optimal-to-lewis-weight-primal} 
For any $v \in \mathbb{R}^m_{>0}$ satisfying $\rho_\max(v)\leq 1+\alphap$ and $\fobj_\primal(v)-\fobj_\primal(v_{*})\leq \tilde\epsilon$ with
\[
\tilde{\epsilon}
= \frac{\alphap^8 \epsilon^4}{\bigl(25 m (\sqrt{n}+\alphap)(\alphap+\alphap^{-1})\bigr)^4},
\]
Then, the vector $\widehat{w}$ defined as $\widehat{w}_i = \bigl(a_i^\top (A^\top V A)^{-1} a_i\bigr)^{1+1/\alphap}$, is an $\epsilon$-estimate of $\sigma_p(A)$.
\end{lemma}

We first prove a key lemma that allows us to compare quadratic forms associated with $\widetilde{w}$ and~$w$.

\begin{lemma}
\label{lem:mult_diff_S}
Let $v\in\R_{>0}^{m}$, $S \subseteq [m]$, and define 
\[
\widetilde v_i = \begin{cases} \left(\sigma_i(v)/v_i\right)^{1/\alphap} & \text{ if } i \in S \\ v_i & \text{ otherwise}. \end{cases}
\]
Then 
\[
(1-\delta) A^\top V A
\preceq A^\top \widetilde V A
\preceq (1+\delta) A^\top V A,
\]
where $\delta \coloneqq \sum_{i \in S} |\widetilde v_i^{1+\alphap} - \sigma_i(v)|$.
\end{lemma}
\begin{proof}
Let $\Delta\defeq \widetilde{v}-v$ and let $\Delta_{+}\defeq\max\{\Delta,\vzero\}$
and $\Delta_{-}\defeq\max\{-\Delta,\vzero\}$ entrywise so that $\Delta_{+},\Delta_{-}\in\R_{\geq0}^{m}$,
$\Delta=\Delta_{+}-\Delta_{-}$. Let $H = A^\top V A$ and $\widetilde H = A^\top \widetilde V A$. Then, we have
\begin{align*}
\normFull{\mh^{-1/2}(\widetilde  H-H)\mh^{-1/2}}_{1} 
&\leq\normFull{\mh^{-1/2}\big(\ma^{\top}\Delta_{+}\ma\big)\mh^{-1/2}}_{1}
+\normFull{\mh^{-1/2}\big(\ma^{\top}\Delta_{-}\ma\big)\mh^{-1/2}}_{1}\,.
\end{align*}
We then use the fact $\Delta_+$ is positive semidefinite, and therefore so is $\mh^{-1/2}\left(\ma^{\top}\Delta_{+}\ma\right)\mh^{-1/2}$, to upper bound the spectral norm by the trace: 
\begin{align*}
\normFull{\mh^{-1/2}\big(\ma^{\top}\Delta_{+}\ma\big)\mh^{-1/2}}_{1}
&=\tr\left[\mh^{-1/2}\big(\ma^{\top}\Delta_{+}\ma\big)\mh^{-1/2}\right]=\sum_{i\in[m]}\left[\Delta_{+}\right]_{i}\left[\ma(\ma^{\top}\mv\ma)^{-1}\ma^{\top}\right]_{ii}\\
& =\sum_{i\in[m]}\left[\Delta_{+}\right]_{i}\cdot\frac{\sigma_i(v)}{v_i}\,.
\end{align*}
By symmetry the same bound holds for $\Delta_{-}$.
Combining the two bounds yields that 
\[
\normFull{\mh^{-1/2}(\ma^{\top}\Delta\ma)\mh^{-1/2}}_1
\leq\sum_{i\in[m]}\left|\Delta_{i}\right|\cdot\frac{\sigma_i(v)}{v_{i}}
\]
Recalling that $\Delta = \widetilde{v}- v$ we obtain the desired upper bound:
\[
\sum_{i\in[m]}\left|\Delta_{i}\right|\cdot\frac{\sigma_i(v)}{v_{i}}
=
\sum_{i \in S} |\widetilde v_i^{1+\alphap} - \sigma_i(v)| = \delta.
\]
Finally, we conclude that
\begin{align*}
    \normFull{\mh^{-1/2}(\ma^{\top}\Delta\ma)\mh^{-1/2}}
    \leq \normFull{\mh^{-1/2}(\ma^{\top}\Delta\ma)\mh^{-1/2}}_1 
    \leq \delta
\end{align*}
and thus $(1-\delta)H
\preceq \widetilde H
\preceq (1+\delta)H$.
\end{proof}

\begin{lemma}\label{lem:function_value_gap_to_small}
    For any $0<\epsilon<1$ and any $v\in\R^m_{>0}$ we have 
    \[
    \sum_{i \in [m]: \rho_i(v) \leq 1-\eps} v_i^{1+\alphap} \leq \frac{6 \baralphap\left(\fobj_\primal(v)-\fobj_\primal(v_*)\right) (1+\rho_\max(v))}{\eps^2}.
    \]
\end{lemma}

\begin{proof}
    This follows immediately from Lemma~\ref{lem:fobj_primal_gap_lower}. 
    Indeed, 
    \begin{align*}
        \sum_{i \in [m]: \rho_i(v) \leq 1-\eps} v_i^{1+\alphap} &\leq \sum_{i \in [m]: \rho_i(v) \leq 1-\eps} v_i^{1+\alphap} \frac{(\rho_i(v)-1)^2}{\eps^2} \frac{\rho_\max(v)+1}{\rho_i(v)+1} \\
        &\leq \frac{\rho_\max(v)+1}{\eps^2} \sum_{i \in [m]} v_i^{1+\alphap} \frac{(\rho_i(v)-1)^2}{\rho_i(v)+1} \\
        &\leq \frac{6 \baralphap\left(\fobj_\primal(v)-\fobj_\primal(v_{*})\right) (1+\rho_\max(v))}{\eps^2},
    \end{align*}
    where the last inequality uses Lemma~\ref{lem:fobj_primal_gap_lower}.
\end{proof}
\begin{lemma} \label{lem:post-processing-error}
Let $v\in \R^m_{>0}$ be such that $\fobj_\primal(v)-\fobj_\primal(v_{*}) \leq \eps^3$ and $\rho_\max(v)\leq 1+\eps$. Let $\widetilde v$ be defined as follows
\[
\widetilde v_i = \begin{cases} \left(\sigma_i(v)/v_i\right)^{1/\alphap} & \text{ if } \rho_i(v) \leq 1-\eps \\ v_i & \text{ otherwise}. \end{cases}
\]
Then $\sum_{i \in [m]: \rho_i(v) \leq 1-\eps} |\widetilde v_i^{1+\alphap} - \sigma_i(v)| \leq 12\baralphap(1+\eps) \eps$.
\end{lemma}
\begin{proof}
Apply the triangle inequality. Observe that for each $i$ for which $\rho_i(v) \leq 1-\eps$ we have: $\widetilde v_i^{1+\alphap} \leq v_i^{1+\alphap}$ since $\left(\sigma_i(v)/v_i\right)^{1/\alphap} = v_i \rho_i(v)^{1/\alphap}$, and $\sigma_i(v) \leq v_i^{1+\alphap}$. Finally, apply Lemma~\ref{lem:function_value_gap_to_small}.     
\end{proof}

\begin{proof}[Proof of Theorem~\ref{thm:approx-min-to-two-sided}]
We first note that by Lemma~\ref{lem:mult_diff_S} and Lemma~\ref{lem:post-processing-error} we have 
\[
(1-\delta) \ma^{\top}V\ma 
\preceq \ma^{\top}\widetilde V\ma
\preceq (1+\delta) \ma^{\top}V\ma
\]
for $\delta = 24\baralphap\eps$. This shows that 
\begin{align} \label{eq:spec}
\frac{1}{1+\delta} \big( \ma^{\top}V\ma \big)^+
\preceq \big( \ma^{\top}\widetilde V\ma \big)^+
\preceq \frac{1}{1-\delta} \big( \ma^{\top}V \ma \big)^+,
\end{align}

We now bound $\rho_i(\widetilde v)$. 
We distinguish two cases: $\rho_i(v) <1-\eps$ and $\rho_i(v) \geq 1-\eps$. First, when $\rho_i(v) <1-\eps$ we have $\widetilde v_i = (\sigma_i(v)/v_i)^{1/\alphap}$, and therefore 
\begin{align*}
\rho_i(\widetilde v)
& =\widetilde v_{i}^{-\alphap}\cdot\left[\ma(\ma^{\top}\widetilde V\ma)^{-1}\ma^{\top}\right]_{ii}
=\frac{\left[\ma(\ma^{\top}\widetilde V\ma)^{+}\ma^{\top}\right]_{ii}}{\left[\ma(\ma^{\top}V\ma)^{+}\ma^{\top}\right]_{ii}}\,.
\end{align*}
Eq.~\eqref{eq:spec} then shows that $\rho_i(\widetilde v) \in [(1+\delta)^{-1},(1-\delta)^{-1}]$. 

Second, when $\rho_i(v) \geq 1-\eps$ we proceed as follows. We have $\widetilde v_i = v_i$ and therefore 
\begin{align*}
\rho_i(\widetilde v) & = v_{i}^{-\alphap}\cdot\left[\ma(\ma^{\top}\widetilde V\ma)^{-1}\ma^{\top}\right]_{ii} = \rho_i(v) \cdot \frac{\left[\ma(\ma^{\top}\widetilde V\ma)^{-1}\ma^{\top}\right]_{ii}}{\left[\ma(\ma^{\top} V\ma)^{-1}\ma^{\top}\right]_{ii}} .
\end{align*}
Eq.~\eqref{eq:spec} then shows that $\rho_i(\widetilde v) \in [(1+\delta)^{-1} \rho_i(v),(1-\delta)^{-1} \rho_i(v)] \subseteq [\frac{1-\eps}{1+\delta},\frac{1+\eps}{1-\delta}]$.

Combining the two cases shows that $\widetilde w=\widetilde v^{1+\alphap}$ is a two-sided $\tilde \eps$-approximation of $\sigma_p(A)$ for $\tilde \eps = \frac{\delta+\eps}{1-\eps} = 25\baralphap\eps/(1-\eps)\leq 50\baralphap\eps$. 
\end{proof}

\subsection{Improved analysis of Lee's algorithm}\label{sec:improved-analysis-Lee's-alg}

In this section, we show how to use \Cref{thm:one-sided-to-two-sided} to prove that a variation of Lee's algorithm~\cite{lee16} computes a two-sided $\epsilon$-approximation of $\sigma_p(A)$ using approximate leverage score computations, at the expense of a $\poly(n,p)$-overhead in precision and a dimension dependent number of iterations, see \cref{thm:algo}.

To prove this result, we first establish a simple lemma that shows that two-sided approximation is ``stable'' with respect to a multiplicative change (i.e., if $w$ is a two-sided approximation then so is its multiplicative approximation).

\begin{lemma} \label{lem:stab}
    Let $\gamma \geq 1$.
    Let $w,\widetilde w \in \R^m_{> 0}$ be such that $\gamma^{-1} \widetilde w_i \leq w_i \leq \gamma \widetilde w_i$ for all $i \in [m]$. Then 
    \[
    \gamma^{-1} \frac{\sigma_i(\widetilde W^{\frac{1}{2}-\frac{1}{p}}\ma)}{\widetilde w_i} \leq \frac{\sigma_i(W^{\frac{1}{2}-\frac{1}{p}}\ma)}{w_i} \leq \gamma \frac{\sigma_i(\widetilde W^{\frac{1}{2}-\frac{1}{p}}\ma)}{\widetilde w_i}
    \]
\end{lemma}
\begin{proof}
We first prove the first inequality. We have that 
\begin{align*}
    \frac{\sigma_i(\widetilde W^{\frac{1}{2}-\frac{1}{p}}\ma)}{\widetilde w_i} &= \widetilde w_{i}^{-\frac{2}{p}}\cdot\left[\ma(\ma^{\top} \widetilde W^{1-\frac{2}{p}}\ma)^{+}\ma^{\top}\right]_{ii}\\
    &\leq \gamma \, w_{i}^{-\frac{2}{p}}\cdot\left[\ma(\ma^{\top}  W^{1-\frac{2}{p}}\ma)^{+} \ma^{\top}\right]_{ii}
    =\gamma \frac{\sigma_i( W^{\frac{1}{2}-\frac{1}{p}}\ma)}{w_i}
\end{align*}
where the inequality uses $\widetilde w_i^{-\frac2p} \leq \gamma^{\frac2p} w_i^{-\frac2p}$ and $(\ma^\top \widetilde W^{1-\frac2p} \ma)^+ \preceq \gamma^{1-\frac{2}{p}} (\ma^\top W^{1-\frac2p} \ma)^+$. The second inequality of the lemma follows by exchanging the roles of $w$ and $\widetilde w$. 
\end{proof}

We can now state our variation of Lee's algorithm and prove its correctness. 

\begin{algorithm}[htbp!]
\caption{Two-sided Lewis weight approximation} \label{alg:low-precision-Lewis}
\Input{$\ma \in \R^{m \times n}$, $p \geq 2$, accuracy $\epsilon>0$}
Let $w_i^{(1)} = n/m$ for all $i \in [m]$, $\eps_1 = \epsilon/(100 p n)$, $\eps_2 = \epsilon/(3p)$, $T = \lceil 2\log(m/n)/\eps_1 \rceil$;

\For{$k=1,\ldots,T-1$}{ 
    Let $w^{(k+1)}$ be $\eps_1/4$-estimates of $\sigma((\mw^{(k)})^{\frac12 - \frac1p} \ma)$\;
}

Let $w  = \frac{1}{T} \sum_{k\in[\totaliters]} w^{(k)}$ and $s$ be $\eps_2$-estimates of $\sigmapw$\;

\Return{$\widetilde w$ with $\widetilde w_i  = w_i (s_i/w_i)^{\frac{p}{2}}$ for all $i \in [m]$}
\end{algorithm}

\begin{theorem}\label{thm:algo} \cref{alg:low-precision-Lewis} outputs a two-sided $\epsilon$-approximation of the $\ell_p$-Lewis weights of $\ma$. Each iteration computes the leverage scores of $DA$ of some diagonal matrix $D$ to multiplicative accuracy $O(\epsilon/(pn))$.
\end{theorem}
\begin{proof}
Steps 1.-4.~of the algorithm correspond to Algorithm 6 by \cite{lee16}.
In Theorem~5.3.4 of~\cite{lee16} it is shown that the resulting $w$ satisfies $w_i/\sigmaipw \geq \exp(-\eps_1)$ and therefore $\sigmaipw \leq \exp(\eps_1) w_i \leq (1+2\eps_1) w_i$.
Moreover, $w$ is an average over $\eps_1/4$-approximate leverage scores so that $\|w\|_1 \leq (1+\eps_1/4)n$, and hence $w$ is a one-sided $2\eps_1$-approximation of $\sigma_p(A)$.
By Theorem~\ref{thm:one-sided-to-two-sided} this implies the vector $\widehat w\coloneqq\sigma(w^{\frac12-\frac1p})^{\frac p2}/w^{\betap}$ is a two-sided Lewis weight approximation with approximation factor 
\[
6\barbetap n\epsilon_1 (1 + 2\epsilon_1)^{\barbetap} n \leq \epsilon/3
\]
by our choice of $\eps_1$.
Finally, we use $\eps_2$-estimates $s$ of $\sigmapw$ to define $\widetilde w_i = w_i (s_i/w_{i})^{\frac{p}{2}}$, so that
\[
(1-\eps_2)^{p/2} \widehat{w}_i
\leq \widetilde{w}_i
\leq (1+\eps_2)^{p/2} \widehat{w}_i
\leq \frac{1}{(1-\eps_2)^{p/2}} \widehat{w}_i.
\]
We can now apply Lemma~\ref{lem:stab} with $\gamma = 1/(1-\eps_2)^{p/2} \leq 1 + \epsilon/3$ by our choice of $\eps_2$. This implies that the~$\widetilde{w}_i$'s are two-sided Lewis weight approximations satisfying
\[
(1 - \epsilon/3)^2
\leq \frac{\sigma_i(\widetilde{W}^{\frac12 - \frac1p} \ma)}{\widetilde w_i}
\leq (1 + \epsilon/3)^2.
\]
Using that $(1-\eps/3)^2 \geq 1-\eps$ and $(1+\eps/3)^2 \leq 1+\eps$, this proves the claim.
\end{proof}
\Cref{thm:algo} implies that we can obtain a two-sided $\eps$-approximation of $\sigma_p(A)$ by iteratively computing $O(p n\log m/\eps)$ many $O(\eps/(pn))$-approximate leverage scores.

\section{Conclusion}

In this paper we provide two algorithms for computing approximations of $\ell_p$-Lewis weights. Additionally, we provide simple procedures that convert weaker notions of approximation into stronger ones, e.g., that turn one-sided approximations into two-sided approximations. For the fundamental problem of computing $\epsilon$-estimates, our methods improve upon the prior state-of-the-art by a factor of $p$. Moreover, we obtain these algorithms by a general \emph{locally} relatively smooth gradient descent method and straightforward applications of it to convex formulations of Lewis weights. 

Altogether, these algorithms and the analysis shed light on the complexity of $\ell_p$-Lewis weight computation, through the lens of relative smoothness and strong convexity. Given the fundamental and pervasive nature of  $\ell_p$-Lewis weights and how natural the associated objective functions are, we hope this work may facilitate the development of efficient optimization algorithms more broadly.

\section*{Acknowledgments}

We thank Simon Apers for many useful discussions during the development of this work. We thank anonymous reviewers from COLT 2026 for their feedback and LLMs for writing advice. Aaron Sidford was supported in part by a Microsoft Research Faculty Fellowship, NSF CAREER Grant CCF1844855, NSF Grant CCF-1955039, and a PayPal research award. Chenyi Zhang was supported by a Shoucheng Zhang Graduate Fellowship. 

\bibliographystyle{alpha}
\bibliography{biblio}

@book{wojtaszczyk1991,
	author = {Wojtaszczyk, P.},
	title = {Banach Spaces for Analysts},
	year = {1991},
	address = {Cambridge University Press}
}

@article{khachiyan1996rounding,
 author = {Leonid G. Khachiyan},
 journal = {Mathematics of Operations Research},
 number = {2},
 pages = {307--320},
 publisher = {INFORMS},
 title = {Rounding of Polytopes in the Real Number Model of Computation},
 volume = {21},
 year = {1996}
}

@incollection{john1948,
    author = {Fritz, John},
    title = {Extremum problems with inequalities as subsidiary conditions},
    booktitle = {Studies and Essays Presented to R. Courant on his 60th Birthday},
    publisher = {Interscience Publishers, New York},
    year = {1948},
}

@INPROCEEDINGS{sidford2018coordinate,
  author={Sidford, Aaron and Tian, Kevin},
  booktitle={2018 IEEE 59th Annual Symposium on Foundations of Computer Science (FOCS)}, 
  title={Coordinate Methods for Accelerating $\ell_\infty$ Regression and Faster Approximate Maximum Flow}, 
  year={2018},
  volume={},
  number={},
  pages={922-933},
  doi={10.1109/FOCS.2018.00091}}

@inproceedings{malitsky2020adaptive,
author = {Malitsky, Yura and Mishchenko, Konstantin},
title = {Adaptive gradient descent without descent},
year = {2020},
publisher = {JMLR.org},
booktitle = {Proceedings of the 37th International Conference on Machine Learning},
articleno = {622},
numpages = {11},
series = {ICML'20}
}

@article{latafat2025,
  title        = {Adaptive proximal algorithms for convex optimization under local Lipschitz continuity of the gradient},
  author       = {Puya Latafat and Andreas Themelis and Lorenzo Stella and Panagiotis Patrinos},
  journal      = {Mathematical Programming},
  volume       = {213},
  pages        = {433--471},
  year         = {2025},
  doi          = {10.1007/s10107-024-02143-7},
}

@article{bauschke2017,
author = {Bauschke, Heinz H. and Bolte, J\'{e}r\^{o}me and Teboulle, Marc},
title = {A Descent Lemma Beyond Lipschitz Gradient Continuity: First-Order Methods Revisited and Applications},
journal = {Mathematics of Operations Research},
volume = {42},
number = {2},
pages = {330-348},
year = {2017},
doi = {10.1287/moor.2016.0817},
}

@inproceedings{BLLSSWW21,
  author       = {Jan van den Brand and
                  Yin Tat Lee and
                  Yang P. Liu and
                  Thatchaphol Saranurak and
                  Aaron Sidford and
                  Zhao Song and
                  Di Wang},
  title        = {Minimum cost flows, MDPs, and $\ell_1$-regression
                  in nearly linear time for dense instances},
  booktitle    = {Proceedings of the fifty-third annual ACM symposium on Theory of Computing},
  publisher    = {{ACM}},
  year         = {2021}
}

@inproceedings{JLS22,
  author       = {Arun Jambulapati and
                  Yang P. Liu and
                  Aaron Sidford},
  title        = {Improved iteration complexities for overconstrained \emph{p}-norm
                  regression},
  booktitle    = {Proceedings of the fifty-fourth annual ACM symposium on Theory of Computing},
  publisher    = {{ACM}},
  year         = {2022}
}

@article{BLM89,
author = {J. Bourgain and J. Lindenstrauss and V. Milman},
title = {{Approximation of zonoids by zonotopes}},
volume = {162},
journal = {Acta Mathematica},
publisher = {Institut Mittag-Leffler},
pages = {73 -- 141},
year = {1989},
doi = {10.1007/BF02392835},
URL = {https://doi.org/10.1007/BF02392835}
}

@inproceedings{fazel22,
author = {Maryam Fazel and Yin Tat Lee and Swati Padmanabhan and Aaron Sidford},
title = {Computing {Lewis} Weights to High Precision},
booktitle = {Proceedings of the 2022 Annual ACM-SIAM Symposium on Discrete Algorithms (SODA)},
chapter = {},
pages = {2723-2742},
doi = {10.1137/1.9781611977073.107},
URL = {https://epubs.siam.org/doi/abs/10.1137/1.9781611977073.107},
year = {2022},
}

@article{Lewis78,
author = {Lewis, D.},
journal = {Studia Mathematica},
number = {2},
pages = {207-212},
title = {{Finite dimensional subspaces of $L_{p}$}},
url = {http://eudml.org/doc/218208},
volume = {63},
year = {1978},
}

@phdthesis{lee16,
  author    = {Yin Tat Lee},
  title     = {Faster Algorithms for Convex and Combinatorial Optimization},
  year      = {2016},
  school    = {Massachusetts Institute of Technology}
}

@inproceedings{cohen2015lp,
  title={$\ell_p$ row sampling by {Lewis} weights},
  author={Cohen, Michael B. and Peng, Richard},
  booktitle={Proceedings of the forty-seventh annual ACM symposium on Theory of Computing},
  pages={183--192},
  year={2015}
}

@article{drineas2012fast,
  title={Fast approximation of matrix coherence and statistical leverage},
  author={Drineas, Petros and Magdon-Ismail, Malik and Mahoney, Michael W. and Woodruff, David P.},
  journal={The Journal of Machine Learning Research},
  volume={13},
  number={1},
  pages={3475--3506},
  year={2012},
  publisher={JMLR. org}
}

@article{rudelson2007sampling,
  title={Sampling from large matrices: An approach through geometric functional analysis},
  author={Rudelson, Mark and Vershynin, Roman},
  journal={Journal of the ACM (JACM)},
  volume={54},
  number={4},
  pages={21--es},
  year={2007},
  publisher={ACM New York, NY, USA}
}

@inproceedings{Woodruff2023online,
author = {David P. Woodruff and Taisuke Yasuda},
title = {Online Lewis Weight Sampling},
booktitle = {Proceedings of the 2023 Annual ACM-SIAM Symposium on Discrete Algorithms (SODA)},
chapter = {},
pages = {4622-4666},
doi = {10.1137/1.9781611977554.ch175},
URL = {https://epubs.siam.org/doi/abs/10.1137/1.9781611977554.ch175},
eprint = {https://epubs.siam.org/doi/pdf/10.1137/1.9781611977554.ch175},
year = {2023},
}

@article{apers2023quantum-journal,
author = {Apers, Simon and Gribling, Sander},
title = {{Quantum Speedups for Linear Programming via Interior Point Methods}},
journal = {SIAM Journal on Computing},
volume = {55},
number = {1},
pages = {93-134},
year = {2026},
URL = {https://doi.org/10.1137/25M1736098}}

@article{apers2024lewis,
  title={{On computing approximate Lewis weights}},
  author={Apers, Simon and Gribling, Sander and Sidford, Aaron},
  journal={arXiv:2404.02881},
  year={2024}
}

@inproceedings{spielman2008graph,
  title={Graph sparsification by effective resistances},
  author={Spielman, Daniel A. and Srivastava, Nikhil},
  booktitle={Proceedings of the fortieth annual ACM symposium on Theory of Computing},
  pages={563--568},
  year={2008}
}

@article{lee2019solving,
  title={Solving linear programs with $\sqrt{\text{rank}}$ linear system solves},
  author={Lee, Yin Tat and Sidford, Aaron},
  journal={arXiv preprint arXiv:1910.08033},
  year={2019}
}

@book{todd2016minimum,
  title={Minimum-volume ellipsoids: Theory and algorithms},
  author={Todd, Michael J.},
  year={2016},
  publisher={SIAM}
}

@article{talagrand1990embedding,
  title={{Embedding subspaces of $L_1$ into $\ell_1^N$}},
  author={Talagrand, Michel},
  journal={Proceedings of the American Mathematical Society},
  volume={108},
  number={2},
  pages={363--369},
  year={1990}
}

@article{lu2018relatively,
  title={Relatively smooth convex optimization by first-order methods, and applications},
  author={Lu, Haihao and Freund, Robert M. and Nesterov, Yurii},
  journal={SIAM Journal on Optimization},
  volume={28},
  number={1},
  pages={333--354},
  year={2018},
  publisher={SIAM},
eprint={arXiv:1610.05708}
}

@misc{tseng2008accelerated,
  title={On accelerated proximal gradient methods for convex-concave optimization},
  author={Tseng, Paul},
  year={2008}
}

@article{clarkson2017low,
  title={Low-rank approximation and regression in input sparsity time},
  author={Clarkson, Kenneth L. and Woodruff, David P.},
  journal={Journal of the ACM (JACM)},
  volume={63},
  number={6},
  pages={1--45},
  year={2017},
  publisher={ACM New York, NY, USA}
}

@misc{li2018general,
  title={A general convergence result for mirror descent with Armijo line search},
  author={Li, Yen-Huan and Riofrio, Carlos A. and Cevher, Volkan},
  note={arXiv:1805.12232},
  year={2018}
}

@article{godeme2023provable,
  title={Provable phase retrieval with mirror descent},
  author={Godeme, Jean-Jacques and Fadili, Jalal and Buet, Xavier and Zerrad, Myriam and Lequime, Michel and Amra, Claude},
  journal={SIAM Journal on Imaging Sciences},
  volume={16},
  number={3},
  pages={1106--1141},
  year={2023},
  publisher={SIAM}
}

\appendix
\section{Technical lemmas}
\label{sec:deferred proofs}

Here we collect several standard estimates used throughout the paper.
\begin{lemma}\label{lem:bregman_divergence_to_spectral_difference}
Let $h:\Sn_{\succ 0} \to \R$ be defined as $h(M) = -\logdet(M)$. Then, for any $M_1,M_2\in\Sn_{\succ 0}$ satisfying $D_h(M_2,M_1)\le \epsilon$ for some $\epsilon\leq 1/10$, we have
\[
(1-2\sqrt{\epsilon})M_1 \preceq M_2 \preceq (1+2\sqrt{\epsilon})M_1.
\]
\end{lemma}

\begin{proof}
Denote $\Delta\coloneqq M_1^{-1/2}M_2M_1^{-1/2}\succ 0$. Then  
\[
D_h(M_2,M_1)
= h(M_2)-h(M_1)-\langle \nabla h(M_1), M_2-M_1\rangle
= \operatorname{Tr}(\Delta)-\log\det \Delta - n.
\]
Hence, letting $\lambda_1,\dots,\lambda_n$ be the eigenvalues of $\Delta$ and $\phi(\lambda)\coloneqq \lambda-\log\lambda-1$, we have  
\[
D_h(M_2,M_1) = \sum_{i \in [n]} \phi(\lambda_i).
\]
Note that the function $\phi$ is convex on $(0,\infty)$, with a unique minimizer at $\lambda=1$ and $\phi(1)=0$. Hence, $\phi(\lambda)\ge 0$ for all $\lambda>0$, which gives $\phi(\lambda_i)\le D_h(M_2,M_1)\le \epsilon\le 1/10$ and thus $|\lambda_i-1|\le 2\sqrt{\epsilon}$ for all $i$. Therefore,
\[
(1-2\sqrt{\epsilon})I \preceq \Delta \preceq (1+2\sqrt{\epsilon})I.
\]
Conjugating by $M_1^{1/2}$ yields
\[
(1-2\sqrt{\epsilon})M_1 \preceq M_2 \preceq (1+2\sqrt{\epsilon})M_1,
\]
which completes the proof.
\end{proof}

\begin{lemma} \label{lem:abs_diff_2}
    Suppose $x,y \in \R^m_{\geq 0}$ and $\delta>0$ are such that $y \leq (1+\delta) x$ entrywise and $\|x\|_1 \leq (1+\delta)\|y\|_1$. Then $\|x-y\|_1 \leq 3 \delta \|y\|_1$.
\end{lemma}
\begin{proof}
The proof follows from writing $x-y = (x-\frac{1}{1+\delta} y) - \frac{\delta}{1+\delta} y$ and applying the triangle inequality:
\begin{align*}
\norm{x-y}_{1} &\leq\sum_{i\in[m]}\left|x_{i}-\frac{1}{1+\delta}y_{i}\right|+ \frac{\delta}{1+\delta} \sum_{i\in[m]}|y_{i}| =\sum_{i\in[m]}x_{i}-\frac{1}{1+\delta}y_{i}+ \frac{\delta}{1+\delta} \sum_{i\in[m]}y_{i} \\
&= \|x\|_1 - \frac{1}{1+\delta} \|y\|_1 + \frac{\delta}{1+\delta} \|y\|_1. 
\end{align*}
Finally, using $\|x\|_1 \leq (1+\delta) \|y\|_1$ we obtain $\norm{x-y}_{1} \leq \left(\frac{(1+\delta)^2 - 1+\delta}{1+\delta}\right) \|y\|_1 \leq 3 \delta \|y\|_1$.
\end{proof}

\begin{lemma}\label{lem:leverage-score-one-norm}
$\Tr[\hat{B} U \hat{B}^\top]\leq \| U \|_1$ for any full column rank $B\in\R^{m\times n}$ and PSD symmetric matrix $ U \in\R^{n\times n}$, where $\hat{B} \defeq B(B^\top B )^{-1/ 2}$. 
\end{lemma}
\begin{proof} Since $U \succeq 0$, we can write 
$\Tr[\hat{ B }U\hat{ B }^\top]=\Tr[\hat{ B }U^{1/2}U^{1/2}\hat{ B }^\top]=\Tr[U^{1/2}\hat{ B }^\top\hat{ B }U^{1/2}]$.
We then note that 
$\hat{ B }^\top\hat{ B }=( B ^\top B )^{-1 / 2} B ^\top B ( B ^\top B )^{-1 / 2} \preceq I$,
which leads to 
$\Tr[\hat{ B }U\hat{ B }^\top]\leq\Tr[U]\leq\|U\|_1$.
\end{proof}

\begin{lemma}\label{lem:matrix-ell1-distance-2}
For any $\zeta>0$ and two full-rank PSD matrices $M_1$ and $M_2$ satisfying
\begin{align}\label{eq:inverse-norm-distance}
\|M_1^{-1/2}(M_2-M_1)M_1^{-1/2}\|_1\leq\zeta\leq\frac{1}{2},
\end{align}
we have
\begin{align*}
\|M_2^{-1/2}(M_1-M_2)M_2^{-1/2}\|_1\leq \frac{\zeta}{1-\zeta}.
\end{align*}
\end{lemma}
\begin{proof}
Let $N \defeq M_1^{-1/2} M_2 M_1^{-1/2}$ 
so that \eqref{eq:inverse-norm-distance} is equivalent to the statement that $\norm{N - I}_1 \leq \zeta\leq 1/2$. Note that
\begin{align*}
\|M_2^{-1/2}(M_1-M_2)M_2^{-1/2}\|_1
= \norm{N^{-1} - I}_1  
= \norm{(N - I) (N)^{-1}}_1
\leq \norm{N - I}_1 \norm{N^{-1}}\,.
\end{align*}
However, since $\norm{N - I} \leq \norm{N - I}_1 \leq \zeta$ we know that every eigenvalue of $N$ is between $1- \zeta$ and $1 + \zeta$. Consequently, $\norm{N^{-1}} \leq (1 - \zeta)^{-1}$ yielding the result.
\end{proof}

\begin{lemma}\label{lem:inverse-ell1-distance}
For any symmetric matrix $M$ satisfying $\|M\|\leq1/2$, we have
\begin{align*}
\big\|\left(I+M\right)^{-1}-I\big\|_1\leq 2\|M\|_1.
\end{align*}
\end{lemma}
\begin{proof}
We use $\lambda_1,\ldots,\lambda_n$ to denote the eigenvalues of $M$. Then we have
\begin{align*}
\left\|\left(I+M\right)^{-1}-I\right\|_1=\sum_{i\in[n]}\left|\frac{1}{1+\lambda_i}-1\right|\leq 2\sum_{i\in[n]}|\lambda_i|=2\|M\|_1,
\end{align*}
where the inequality is due to the fact that for each $\lambda_i$ we have $|\lambda_i|\leq \|M\|\leq 1/2$.
\end{proof}

\end{document}